\documentclass[journal]{IEEEtran}

\usepackage{fullpage,graphicx,url,setspace}

\usepackage{xparse}
\usepackage{cite}
\usepackage{bm}
\usepackage{color}
\usepackage[small,bf]{caption}
\usepackage{amsmath,verbatim,amssymb,amsfonts,amscd, graphicx, algorithm, algorithmic}
\usepackage{amsmath,amsthm}
\usepackage{mathtools}
\mathtoolsset{mathic}
\usepackage{microtype}
\usepackage{wrapfig}

\usepackage{ifxetex,ifluatex}
\ifxetex
\usepackage{xltxtra} 
\else
\ifluatex
\usepackage{luatextra}
\else
\usepackage{amssymb}
\usepackage{algorithmic, algorithm}
\usepackage[utf8]{inputenc}

\csname else\expandafter\endcsname
\romannumeral -`0
\fi
\fi
\iftrue\relax%
\defaultfontfeatures{Ligatures=TeX}
\usepackage[math-style=TeX, vargreek-shape=TeX]{unicode-math}
\fi

\NewDocumentCommand{\entropy}{om}{\mathbb{H}\left[#2
    \IfValueT{#1}{\,\middle|\,#1}\right]}
\NewDocumentCommand{\bentropy}{lm}
  {\widetilde{\mathbb{H}}#1\left[#2\right]}

\NewDocumentCommand{\mutualInfo}{omm}{\mathbb{I}\left[#2;#3
    \IfValueT{#1}{\,\middle|\,#1}\right]}

\usepackage{tikz}
\usetikzlibrary{shapes.geometric, positioning, shadows}
\usepackage{graphicx,color}

\newtheorem{theorem}{Theorem}

\newtheorem{lemma}{Lemma}

\newtheorem{remark}{Remark}

\DeclareMathOperator*{\argmin}{arg min}
\DeclareMathOperator*{\argmax}{arg max}

\newcommand*{\transpose}[1]{#1\sp{\intercal}}

\newcommand{\chin}{\chi_{n, p, \varepsilon}}

\begin{document}
%
\title{Sequential Information Guided Sensing}


\author{Ruiyang Song,\thanks{Ruiyang Song
    (Email: {songry12@mails.tsinghua.edu.cn}) is with the Dept. of Electronic Engineering, 
Tsinghua University.  Yao Xie (Email: {yao.xie@isye.gatech.edu}) and  
    Sebastian Pokutta
    (Email: {sebastian.pokutta@isye.gatech.edu})
    are with the H. Milton Stewart School of
    Industrial and Systems Engineering, Georgia Institute of
    Technology, Atlanta, GA.}\quad
  \and Yao Xie, \and \quad Sebastian Pokutta
  \thanks{This work is partially supported by an NSF CAREER Award CMMI-1452463
and an NSF grant CCF-1442635. Ruiyang Song was visiting the H. Milton Stewart School of Industrial and Systems Engineering at the Georgia Institute of Technology while working on this paper.   
}
} \date{\today}

\IEEEtitleabstractindextext{%
\begin{abstract}

We study the value of information in sequential compressed sensing by 
characterizing the performance of sequential information guided sensing in practical scenarios when information is inaccurate. In particular, we assume the signal distribution is parameterized through Gaussian or Gaussian mixtures with estimated mean and covariance matrices, and we can measure compressively through a noisy linear projection or using one-sparse vectors, i.e., observing one entry of the signal each time. We establish a set of performance bounds for the bias and variance of the signal estimator via posterior mean, by capturing the conditional entropy (which is also related to the size of the uncertainty), and the additional power required due to inaccurate information to reach a desired precision. Based on this, we further study how to estimate covariance based on direct samples or covariance sketching. Numerical examples also demonstrate the superior performance of Info-Greedy Sensing algorithms compared with their random and non-adaptive counterparts.

\end{abstract}

}

\maketitle

\IEEEdisplaynontitleabstractindextext

%
\IEEEpeerreviewmaketitle

\IEEEPARstart{}{} 

\section{Introduction}

Sequential compressed sensing is a promising new information acquisition and recovery technique to process big data that arises in various applications such as compressive imaging \cite{NeifeldTaskSpecific2008, KeAshok2010, NeifeldCSImaging2011}, 
 power network monitoring
  \cite{WirelessHouseElectricity2014}, and large scale sensor networks
  \cite{sparseSensorLocal2011}. The sequential nature of the problems is either because the measurements are taken one after another, or due to the fact that the data is obtained in a streaming fashion so that it has to be processed in one pass. 
  
To harvest the benefits of adaptivity in sequential compressed sensing,  various
algorithms have been developed (see \cite{InfoGreedy2014} for a  review.) We may classify these algorithms as (1) being agnostic about the signal distribution and, hence, using random measurements 
\cite{HauptAdaptiveCS2009, DavenportArias-Castro2012,TajerPoor2012,MalioutovSanghaviWillsky2010,HauptSeqCS2012,JainSoniHaupt2013,MalloyNowak2013}; 
(2) exploiting additional structure of the signal (such as graphical structure
\cite{KrishnamurthySingh13} and tree-sparse structure \cite{ErvinCastro2013, AkshayHaupt2014}) to design measurements; 
(3) exploiting the distributional information of the signal in choosing the measurements possibly through maximizing mutual information:  the seminal Bayesian compressive
  sensing work \cite{BayesianCS2008}, Gaussian mixture models (GMM) \cite{TaskDrivenDuarte2013, CarsonChenRodrigues2012}, and our earlier work \cite{InfoGreedy2014} which presents a general framework for information guided sensing referred to as \emph{Info-Greedy Sensing}. 

Such additional knowledge about signal structure or distributions are various forms of \emph{information} about the unknown signal. Information may play a distinguishing role: as the compressive imaging example demonstrated in Fig. \ref{high_rel_gt} (see Section \ref{sec:num_egs} for more details), with a bit of (albeit inaccurate) information estimated via random samples of small patches of the image, Info-Greedy Sensing is able to recover details of a high-resolution image, whereas random measurements completely miss the image. 

In this paper we examine the value of information in sequential compressed sensing by considering Info-Greedy Sensing when information is imprecise. Info-Greedy Sensing is a framework introduced in \cite{InfoGreedy2014} that aims at designing subsequent measurements to maximize the mutual information conditioned on previous measurements. 
Conditional mutual information is a natural metric here, as it captures exclusively useful new information between the signal and the results of the measurements disregarding noise and what has already been learned from previous measurements.  We assume information is parameterized imperfectly and captured through sample estimates or ``sketching'', and when measurements of the unknown signal are compressive or even one-sparse (we are only able to inspect one entry of the signal). 
As shown in \cite{InfoGreedy2014}, Info-Greedy Sensing for a Gaussian signal becomes a simple iterative algorithm: choosing the measurement as the leading eigenvector of the conditional signal covariance matrix in that iteration, and then update the covariance matrix via a simple rank-one update, or, equivalently, choosing measurement vectors $a_1, a_2, \ldots$ as the orthonormal eigenvectors of the signal covariance matrix $\Sigma$ in a decreasing order of eigenvalues. This can also be easily generalized to GMM signals, where a heuristic that works well is to measure in the dominant eigenvector direction of the Gaussian component with the highest posterior weight in that iteration. 

In practice, we may be able to estimate the signal covariance matrix to initialize the algorithm through a training session. For Gaussian signals, there are two possible approaches: either using training samples that are sampled from the same distribution, or through the so-called ``covariance sketching'' \cite{CovSketching2012,CovSketchSparse,ChenChiGoldsmith2014}  based on low-dimensional random sketches of the samples. As a consequence, the measurement vectors are calculated from eigenvectors of the estimated covariance matrix $\widehat{\Sigma}$, which deviates from the optimal directions. Since we almost always have to use an estimate for the signal covariance, it is crucial to quantify the performance of sensing algorithms with model mismatch and shed some light on how to properly initialize the algorithm. 

In this paper we characterize the performance of Info-Greedy Sensing for Gaussian and GMM signals (with possibly low-rank covariance matrices) when the true signal covariance matrix is replaced with a proxy, which may be an estimate from direct samples or using a covariance sketching scheme. We establish a set of theoretical results including (1) studying the bias and variance of the signal estimator via posterior mean, by relating the error in the covariance matrix $\|\Sigma -\widehat{\Sigma}\|$ to the entropy of the signal posterior distribution after each sequential measurement, (2)  establishing an upper bound on the additional power needed to achieve the signal precision $\|x-\hat{x}\|\leq \varepsilon$; and (3) translate these into requirements on the choice of the sample covariance matrix through direct estimation or through covariance sketching. Furthermore, we also study Info-Greedy Sensing in a special setting when the measurement vector is desired to be one-sparse, and establish analogously a set of theoretical results. Such a requirement arises from applications such as nondestructive testing (NDT) \cite{NDE2003} or network tomography.   
We also present numerical examples to demonstrate the superior performance of Info-Greedy Sensing compared to a batch method (where measurements are not adaptive) when there is mismatch. 

Our notations are standard. 
Denote $[n] \triangleq \{1,2,\ldots,n\}$; $\|X\|$, $\|X\|_F$, and $\|X\|_{*}$  represent  the spectral norm, the Frobenius  norm, and the nuclear norm of a matrix $X$, respectively; let $\nu_i(\Sigma)$ denote the $i$th largest eigenvalue of a positive semi-definite matrix $\Sigma$; $\|x\|_0$, $\|x\|_1$, and $\|x\|$ represent the $\ell_0$, $\ell_1$ and $\ell_2$ norm of a vector $x$, respectively; let $\chi_n^2$ be the quantile function of the chi-squared distribution with $n$ degrees of freedom; let $\mathbb{E}[x]$ and $\mbox{Var}[x]$ denote the mean and the variance of a random variable $x$; we write $X \succeq 0$ to indicate that the matrix is positive semi-definite; $\phi(x|\mu, \Sigma)$ denotes the probability density function of the multi-variate Gaussian with mean $\mu$ and covariance matrix $\Sigma$; let $e_j$ denote the $j$th column of identity matrix $I$ (i.e., $e_j$ is a vector with only one non-zero entry at location $j$); and $(x)^+ = \max\{x, 0\}$ for $x\in \mathbb{R}$.

\begin{figure}[h!]
\begin{center}
\includegraphics[width = 1\linewidth]{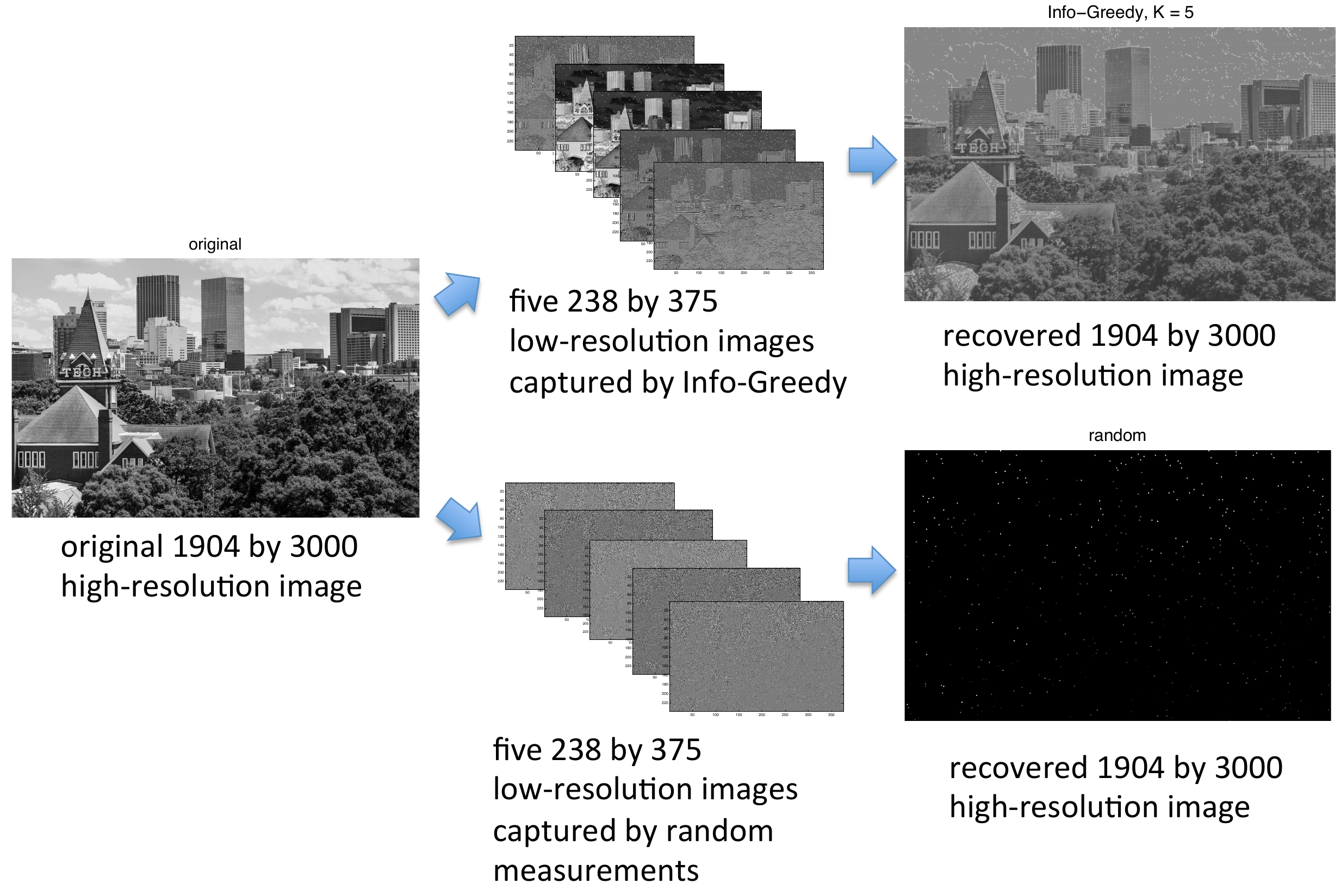} 
\end{center}
\caption{Value of information in sensing a high-resolution image of size 1904$\times$3000. Here, compressive linear measurements correspond to extracting the so-called \emph{features} in compressive imaging \cite{NeifeldTaskSpecific2008, KeAshok2010, NeifeldCSImaging2011}. In this example, the compressive imaging system captures 5 low resolution images of size 238-by-275 using masks designed by Info-Greedy Sensing or random masks 
(this corresponds to compressing the data into $8.32\%$ of its original dimensionality). Info-Greedy Sensing performs much better than random features and preserves richer details in the recovered image. Details are explained in Section \ref{GT_img}.
}
\label{high_rel_gt}
\end{figure}

\section{Info-Greedy Sensing}

A typical sequential compressed sensing setup is as follows. Let $x \in \mathbb{R}^n$ be an unknown $n$-dimensional signal. We make $K$ measurements of $x$ sequentially
\[
y_k = \transpose a_k x + w_k, \quad k = 1, \ldots, K,
\]
and the power of the measurement is $\|a_k\|^2 = \beta_k$. The goal is to recover $x$ using measurements $\{y_k\}_{k=1}^K$. 
Consider a Gaussian signal $x \sim \mathcal{N}(0, \Sigma)$ with known zero mean and covariance matrix $\Sigma$ (here without loss of generality we have assumed the signal has zero mean). Assume the rank of $\Sigma$ is $s$ and the signal can be low-rank $s\ll n$ (however, the algorithm does not require the covariance to be low-rank). 
Info-Greedy Sensing \cite{InfoGreedy2014} chooses each measurement to maximize the conditional mutual information 
\begin{equation}
a_{k} \leftarrow \argmax_{a}
  \mutualInfo[y_{j}, a_{j}, j < k]{x}
  {\transpose{a} x + w}/\beta_k.
  \label{info-greedy}
\end{equation}
The goal is to use a minimum number of measurements (or total power) so that the estimated signal is recovered with precision $\varepsilon$; i.e., $\|\widehat{x} - x\| < \varepsilon$ with a high probability $p$.  Define \[\chi_{n, p, \varepsilon} \triangleq \varepsilon^2/\chi_n^2(p),\] and we will show in the following that this is a fundamental quantity that determines the termination condition of our algorithm to achieve the precision $\varepsilon$ with the confidence level $p$. Note that $\chi_{n, p, \varepsilon}$ is a precision $\varepsilon$ adjusted by the confidence level.

\subsection{Gaussian signal}

In \cite{InfoGreedy2014}, we have devised a solution to (\ref{info-greedy}) when the signal is Gaussian. The measurement will be made in the directions of the eigenvectors of $\Sigma$ in a decreasing order of eigenvalues, and the powers (or the number of measurements) will be such that the eigenvalues  after the measurements are sufficiently small (i.e., less than $\varepsilon$). The power allocation depends on the noise variance, signal recovery precision $\varepsilon$, and confidence level $p$, as given in Algorithm \ref{alg:Gaussian-all}.  

\begin{algorithm}[h!]
  \caption{Info-Greedy Sensing for Gaussian signals}
  \begin{algorithmic}[1]
    \REQUIRE assumed signal mean \(\mu\) and covariance matrix \(\Sigma\), noise variance \(\sigma^2\), recovery accuracy \(\varepsilon\),
      confidence level \(p\)
    \REPEAT 
    \STATE \((\lambda, u) \leftarrow\)
        largest eigenvalue and associated normalized eigenvector of \(\Sigma\)
             \STATE \(\beta \leftarrow \sigma^2 (1/\chi_{n,p,\varepsilon}-1/\lambda)^+\)
      
      \STATE $a = \sqrt{\beta} u$, \(y = \transpose a x + w\) \COMMENT{measure}

     \STATE \(\mu \leftarrow \mu
        + \Sigma a
        (y - \transpose a\mu )/(\beta\lambda +\sigma^2)\) \COMMENT{mean}
        
\STATE  \(\Sigma \leftarrow \Sigma
        - \Sigma a
        \transpose a\Sigma/(\beta\lambda + \sigma^2) \) \COMMENT{covariance}
        
          \UNTIL{\(\Vert{\Sigma}\Vert \leq \chin\)}
    \COMMENT{all eigenvalues small}
    \RETURN signal estimate $\hat{x} = \mu$
  \end{algorithmic}
  \label{alg:Gaussian-all}
\end{algorithm}

\subsection{Gaussian mixture model (GMM) signals}

The probability density function of GMM is given by 
\begin{equation*}
p(x) = \sum_{c=1}^C \pi_c \phi(x|\mu_c, \Sigma_c), \label{GMM_model}
\end{equation*}
where $C$ is the number of classes,
and $\pi_c$ is the probability that sample is drawn from class $c$. 
Unlike for Gaussian signals, the mutual information of GMM has no explicit form. However, for GMM signals, there are two approaches that tend to work well: Info-Greedy Sensing derived based on a gradient descent approach \cite{CarsonChenRodrigues2012,InfoGreedy2014} uses the fact that the gradient of the conditional mutual information with respect to $a$ is a linear transform of the minimum mean square error (MMSE)  matrix \cite{PalomarVerdu2006, PayaroPalomar2009}, and the so-called  \emph{greedy heuristic} which approximately maximizes the mutual information. The greedy heuristic picks the Gaussian component with the highest posterior $\pi_c$ at that moment, and chooses the next measurement $a$ to be its eigenvector associated with the maximum eigenvalue, as summarized in Algorithm \ref{alg:GMM_heuristics}. The greedy heuristic can be implemented more efficiently compared to the gradient descent approach and sometimes have competitive performance (see, e.g. \cite{InfoGreedy2014}). 

\begin{algorithm}
  \caption{Heuristic Info-Greedy Sensing for GMM signals}
  \begin{algorithmic}[1]
    \REQUIRE number of components $C$, assumed means \(\{\mu_c\}\), covariances \(\{\Sigma_c\}\), initial weights \(\{\pi_c\}\), 
      noise variance \(\sigma^2\),
      confidence level \(p\)
    \REPEAT 
      \STATE \(c^* \leftarrow \arg\max_c \pi_c \)         
             \STATE \((\lambda, u) \leftarrow\)
        largest eigenvalue and associated normalized eigenvector of \(\Sigma_{c^*}\)
                     \STATE \(\beta \leftarrow \sigma^2 (1/\chi_{n, p, \varepsilon}-1/\lambda)^+\)
\STATE $a = \sqrt{\beta} u$, \(y = \transpose{a} x + w \) \COMMENT{measure}
         \FOR{$c = 1, \ldots, C$}
\STATE \(\mu_c \leftarrow \mu_c
        + 
        [(y - \transpose{a}\mu_c)/(\transpose{a} \Sigma_{c} a + \sigma^{2})]\Sigma_{c} a\)
\STATE \(\Sigma_{c}\leftarrow \Sigma_{c}
        - \Sigma_{c}a
        \transpose{a}\Sigma_{c}/ (\transpose{a} \Sigma_{c} a + \sigma^{2})\)  	\STATE \(
\pi_c \leftarrow K \pi_c  \exp\{- \frac{1}{2} (y - 
\transpose{a} \mu_{c})^2/(\transpose{a}
  \Sigma_{c}a + \sigma^2) \} \)
  \STATE ($K$: normalizing constant)
  \ENDFOR
    \UNTIL{\(\Vert{\Sigma_{c^*}}\Vert \leq \chi_{n, p, \varepsilon}\)
    }
    \RETURN signal class $c^* = \arg\max_c \pi_c$, estimate \(\hat{x}=\mu_{c^*}\)
  \end{algorithmic}
  \label{alg:GMM_heuristics}
\end{algorithm}

\subsection{One-sparse measurement}

The problem of Info-Greedy Sensing with sparse measurement constraint, i.e., each measurement has only $k_0$ non-zero entries $\|a\|_0=k_0$, has been examined in \cite{InfoGreedy2014} and solved using outer approximation (cutting planes). Here we will focus on one-sparse measurements, $\|a\|_0=1$, as it is an important instance arising in applications such as nondestructive testing (NDT). 

\begin{algorithm}[h!]
  \caption{Info-Greedy Sensing with sparse measurement $\|a\|_0=1$, for Gaussian signals}
  \begin{algorithmic}[1]
    \REQUIRE assumed signal mean \(\mu\) and covariance matrix \(\Sigma\), noise variance \(\sigma^2\), recovery accuracy \(\varepsilon\),
      confidence level \(p\)
    \REPEAT 
    \STATE \(j^* \leftarrow \arg\max_j \Sigma_{jj}\)
    \STATE \(a \leftarrow \sqrt{\beta} e_{j^*}\), \(y = \transpose a x + w\) \COMMENT{measure}

     \STATE \(\mu \leftarrow \mu
        + \Sigma a
        (y - \transpose a\mu )/(\beta\Sigma_{j^*j^*} +\sigma^2)\) \COMMENT{mean}
        
\STATE  \(\Sigma \leftarrow \Sigma
        - \Sigma a
        \transpose a\Sigma/(\beta\Sigma_{j^*j^*}  + \sigma^2) \) \COMMENT{covariance}
        
          \UNTIL{\(\Vert{\Sigma}\Vert \leq \chi_{n, p, \varepsilon}\)}
    \COMMENT{all eigenvalues  small}
    \RETURN signal estimate $\hat{x} = \mu$
  \end{algorithmic}
  \label{alg:one-sparse}
\end{algorithm}

Info-Greedy Sensing with one-sparse measurements can be readily derived. Note that the mutual information between $x$ and the outcome using one-sparse measurement $y_1 = \transpose{e_j}x + w_1$ is given by
\[
\mathbb I [x;y_1]=\frac{1}{2}{\rm ln} (\Sigma_{jj}/\sigma^2+1),
\]
where $\Sigma_{jj}$ denote the $j$th diagonal entry of matrix $\Sigma$. Hence, the measurement that maximizes the mutual information is given by $e_{j^*}$ where $j^* = \arg\max_j \Sigma_{jj}$, i.e., measuring in the signal coordinate with the largest variance or largest uncertainty. Then Info-Greedy Sensing measurements can be found iteratively, as presented in Algorithm \ref{alg:one-sparse}. Note that the correlation of signal coordinates are reflected in the update of the covariance matrix: if the $i$th and $j$th coordinates of the signal are highly correlated, then the uncertainty in $j$ will also be greatly reduced if we measure in $i$. 
 A similar algorithm with one-sparse measurement for GMM signals can be derived, where in each iteration we select the component with the largest weight and measure in the signal coordinate with largest variance. 

\subsection{Updating covariance with sequential data}

If our goal is to estimate a sequence of data $x_1, x_2, \ldots$ (versus just estimating a single instance), we may be able to update the covariance matrix using the already estimated signals simply via
\begin{equation}
\widehat{\Sigma}_{t} = \alpha \widehat{\Sigma}_{t-1} + (1-\alpha)\hat{x}_t\hat{x}_t^\intercal, \quad t = 1, 2, \ldots, \label{seq_update_cov}
\end{equation}
and the initial covariance matrix is specified by our prior knowledge
$\widehat{\Sigma}_0 = \widehat{\Sigma}$. Using the updated covariance matrix $\widehat{\Sigma}_t$, we design the next measurement for signal $x_{t+1}$. This way we may be able to correct the inaccuracy of $\widehat{\Sigma}$ by including new samples. We refer to this method as ``Info-Greedy-2'' hereafter.

\section{Performance bounds}

In the following, we establish performance bounds, for cases when we (1) sense Gaussian and GMM signals using estimated covariance matrices; (2) sense Gaussian signals with one-sparse measurements.

\subsection{Gaussian case with model mismatch}

To analyze the performance of our algorithms when the assumed covariance $\widehat{\Sigma}$ used in Algorithm \ref{alg:Gaussian-all} is different from the true signal covariance matrix $\Sigma$, we introduce the following notations. 
Let the eigenpairs of $\Sigma$ with the eigenvalues (which can be zero) ranked from the largest to the smallest to be $(\lambda_1, u_1), (\lambda_2, u_2), \ldots, (\lambda_n, u_n)$, and let the eigenpairs of $\widehat{\Sigma}$  with the eigenvalues  (which can be zero) ranked from the largest to the smallest to be $(\hat{\lambda}_1, \hat{u}_1), (\hat{\lambda}_2, \hat{u}_2), \ldots, (\hat{\lambda}_n, \hat{u}_n)$. Let the updated covariance matrix in Algorithm \ref{alg:Gaussian-all} starting from $\widehat{\Sigma}$ after $k$ measurements be $\widehat{\Sigma}_k$,  and the true posterior covariance matrix of the signal conditioned on these measurements be $\Sigma_k$. The relations of these notations are illustrated in Fig. \ref{evol}. 

\begin{figure}
\begin{center}
\includegraphics[width = 0.7\linewidth]{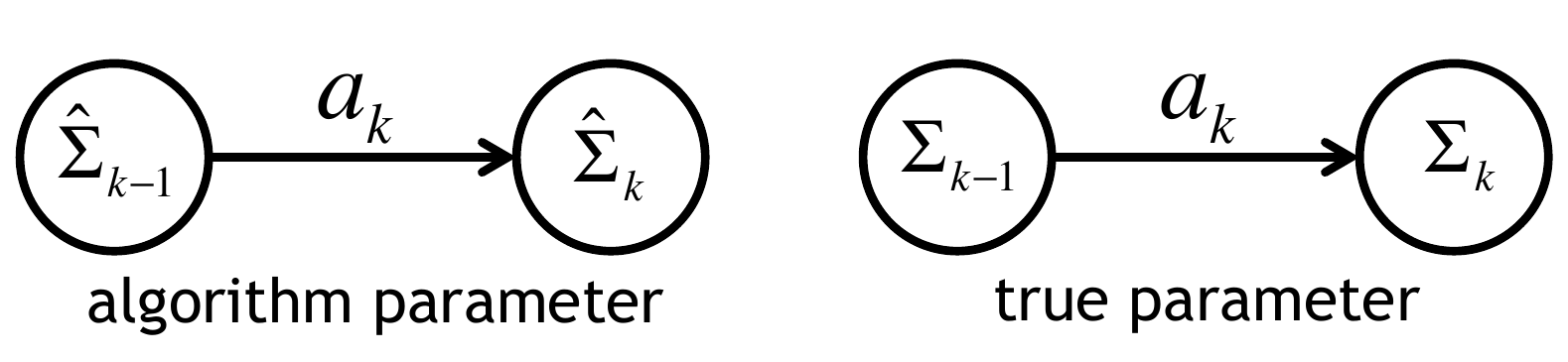}
\caption{Parameter updates performed by the algorithm and updates happen on the true distribution.}
\label{evol}
\end{center}\vspace{-0.2in}

\end{figure}

Note that since each time we measure in the direction of the dominating eigenvector of the posterior covariance matrix, $(\hat{\lambda}_k, \hat{u}_k)$ and $(\lambda_k, u_k)$ correspond to the largest eigenpair of $\widehat{\Sigma}_{k-1}$ and $\Sigma_{k-1}$, respectively.
Furthermore, define the difference between the true and the assumed conditional covariance matrices after $k$ measurements as
\[E_k=\widehat{\Sigma}_k-\Sigma_k,\quad k = 1, \ldots, K,\]
and their sizes  
\[\delta_k=\Vert E_k\Vert, \quad k = 1, \ldots, K.\] 
Let the eigenvalues of $E_k$ be $e_1\geq e_2 \geq\cdots\geq e_n$; then $\delta_k=\max\{\vert e_1 \vert, \vert e_n \vert\}$. Let 
\[
\delta_0 = \|\widehat{\Sigma} - \Sigma\|
\]
denote the size of the initial mismatch.

\subsubsection{Deterministic mismatch} 
First we assume the mismatch is deterministic, and find bounds for bias and variance of the estimated signal. Assume the initial mean is $\hat{\mu}$ and the true signal mean is $\mu$, the updated mean using Algorithm \ref{alg:Gaussian-all} after $k$ measurements is $\hat{\mu}_k$, and the true posterior mean is  $\mu_k$.
\begin{theorem}[Unbiasedness] 
\label{lemma_bias}  After $k$ measurements, the expected difference between the updated mean and the true posterior mean is given by
\[\mathbb E[\hat{\mu}_k - \mu_k]=(\hat{\mu} - \mu) \cdot\prod_{j=1}^k(I_n-\frac{\beta_j\hat{\lambda}_j }{\beta_j\hat{\lambda}_j+\sigma^2} \hat{u}_j \hat{u}_j^{\intercal}).\]
Moreover, if $\hat{\mu} = \mu$, i.e., the assumed mean is accurate, the estimator is unbiased throughout all the iterations $\mathbb E[\hat{\mu}_k - \mu_k]=0$, for $k=1, \ldots, K$.
\end{theorem}

Next we show that the variance of the estimator, when the initial mismatch $\|\widehat{\Sigma} - \Sigma\|$ is sufficiently small, reduces gracefully. This is captured through the reduction of entropy, which is also a measure of the uncertainty in the estimator. In particular, we consider the posterior entropy of the signal conditioned on the previous measurement outcomes. Since the entropy of a Gaussian signal $x \sim\mathcal{N}(\mu, \Sigma)$ is given by
$
  \entropy{x} = \ln \left[(2\pi e)^{n/2} \det^{1/2}(\Sigma)\right],
$
the conditional mutual information is the log of the determinant of the conditional covariance matrix, or equivalently the log of the volume of the ellipsoid defined by the covariance matrix. 
Here, to accommodate the scenario where the covariance matrix is low-rank (our earlier assumption), we consider a modified definition for conditional entropy, which is the log of the volume of the ellipsoid on the low-dimensional space that the signal lies on:
\[
\entropy[y_{j}, a_{j}, j \leq k]{x}
= \ln [(2\pi e)^{s/2} {\textsf{Vol}}(\Sigma_k)],
\] 
where ${\textsf{Vol}}(\Sigma_k)$ is the volume of the ellipse defined by $\Sigma_k$  equal to the product of its non-zero eigenvalues. 
%
%
\begin{theorem}[Entropy of estimator]\label{entropy}
If for some constant $\delta\in (0, 1)$ the error satisfies 
\[
\|\widehat{\Sigma}-\Sigma \|\leq \frac{\delta}{4^{K+1}}\chin,\]
then for $k = 1, \ldots, K$,
\begin{equation}
\entropy[y_{j}, a_{j}, j \leq k]{x} \leq \frac{s}{2}
\left\{\ln [2\pi e~{\rm tr}(\Sigma)]
- \sum_{j=1}^{k}\ln (1/f_j)
\right\}, \label{entropy_bound}
\end{equation}
where
\begin{equation}
f_k=1-\frac{1-\delta}{s}\frac{\beta_k\hat{\lambda}_k}{\beta_k\hat{\lambda}_{k}+\sigma^2} \in (0, 1),\quad k = 1, \ldots, K. \label{def_f}
\end{equation}
\end{theorem}
In the proof of Theorem \ref{entropy}, we use the trace of the underlying actual covariance matrix ${\rm tr}(\Sigma_k)$ as potential function, which serves as a surrogate for the product of eigenvalues that determines the volume of the ellipsoid and hence the entropy, since it is much easier to calculate the trace of the observed covariance matrix ${\rm tr} (\widehat{\Sigma}_{k})$. The following recursion is crucial for the derivation: for an assumed covariance matrix $\Sigma$, after measuring in the direction of a unit norm eigenvector $u$ with eigenvalue \(\lambda\)  using power \(\beta\),
the updated matrix takes the form of
\begin{equation}
  \label{eq:covariance-eigenvector}
  \begin{split}
  &\Sigma - \Sigma \sqrt{\beta} u
  \left(
    \transpose{\sqrt{\beta} u} \Sigma \sqrt{\beta} u
    + \sigma^{2}
  \right)^{-1}
  \transpose{\sqrt{\beta} u} \Sigma \\
&  =
  \frac{\lambda \sigma^{2}}{\beta \lambda + \sigma^{2}}
  u \transpose{u}
  + \Sigma^{\perp u},
  \end{split} 
\end{equation}
where \(\Sigma^{\perp u}\) is the component of \(\Sigma\)
in the orthogonal complement of \(u\).
Thus, the only change in the eigen-decomposition of \(\Sigma\)
is the update of the eigenvalue of \(u\)
from \(\lambda\) to
\(\lambda \sigma^{2} / (\beta \lambda + \sigma^{2})\).
Based on (\ref{eq:covariance-eigenvector}), after one measurement, the trace of the covariance matrix  becomes 
\begin{equation}
{\rm tr}(\widehat{\Sigma}_{k})={\rm tr}(\widehat{\Sigma}_{k-1})-\frac{\beta_{k}\hat{\lambda}_{k}^2}{\beta_{k}\hat{\lambda}_{k}+\sigma^2}.
\label{tr_hat_recursion}
\end{equation}


\begin{remark}
The upper bound of the posterior signal entropy in (\ref{entropy_bound}) shows that the amount of uncertainty reduction by the $k$th measurement is roughly $(s/2) \ln (1/f_k)$.
\end{remark}
\begin{remark}
Use the inequality $\ln(1-x) \leq -x$ for $x\in (0, 1)$, we have that in (\ref{entropy_bound})
\begin{align*}
&\entropy[y_{j}, a_{j}, j \leq k]{x}
\leq  \frac{s}{2}\ln[2\pi e{\rm tr}(\Sigma)]-\frac{1-\delta}{2}\sum_{j=1}^k \frac{\beta_j\hat{\lambda}_j}{\beta_{j}\hat{\lambda}_j+\sigma^2}\\
&=\frac{s}{2}\ln[2\pi e{\rm tr}(\Sigma)]-\frac{k(1-\delta)}{2} +\frac{(1-\delta)}{2}\sum_{j=1}^k\frac{\chin}{\hat{\lambda}_j}.
\end{align*}
On the other hand, in the ideal case if the true covariance matrix is used, the posterior entropy of the signal is given by 
\begin{align}
\mathbb{H}_{\rm ideal}\left[x,\middle| y_{j}, a_{j}, j \leq k\right]
&=\frac{1}{2} \ln [(2\pi e)^s \prod_{j=1}^s \lambda_j ]
-\frac{1}{2}\sum_{j=1}^k \frac{\lambda_j}{\chin},
\label{entropy_mismatch}
\end{align}
where $\tilde{\beta}_j = (1/\chin-1/\lambda_j)^+\sigma^2$. Hence, we have
\begin{align}
&\entropy[y_{j}, a_{j}, j \leq k]{x}  \nonumber\\
& 
~~~~~~\leq \mathbb{H}_{\rm ideal}\left[x,\middle| y_{j}, a_{j}, j \leq k\right]
+ C \nonumber \\
& ~~~~~~~~~- \frac{1}{2} \sum_{j=1}^k \left[
\frac{\lambda_j}{\chin}
 + (1-\delta)\left(1-\frac{\chin}{\hat{\lambda}_j}\right)
\right].
\label{entropy_mismatch_2}
\end{align}
where $C = \frac{s}{2} \ln [{\rm tr}(\Sigma)/\sqrt[s]{\prod_{j=1}^s\lambda_j}]$ is a constant independent of  measurements.  This upper bound has a nice interpretation: it characterizes the amount of uncertainty reduction with each measurement. For example, when the number of measurements required when using the assumed covariance matrix versus using the true covariance matrix are the same, we have $\lambda_j \geq \chin$ and $\hat{\lambda}_j \geq \chin$. Hence, the third term in (\ref{entropy_mismatch_2}) is upper bounded by $- k/2$, which means that the amount of reduction in entropy is roughly 1/2 nat per measurement. 
\end{remark}

\begin{remark}
Consider the special case where the errors only occur in the eigenvalues of the matrix but not in the eigenspace $U$, i.e., 
$\widehat{\Sigma} - \Sigma = U \mbox{diag}\{e_1, \cdots, e_s\} \transpose U$ 
and $\max_{1\leq j\leq s}{\vert e_j \vert}=\delta_0$, 
then the upper bound in (\ref{entropy_mismatch}) can be further simplified. 
Suppose only the first $K$ $(K\leq s)$ largest eigenvalues of $\widehat{\Sigma}$ are larger than the stopping criterion $\chin$ required by the precision, i.e., the algorithm takes $K$ iterations in total. Then  
\begin{align*}
\entropy[y_{j}, a_{j}, j \leq k]{x}
&\leq \mathbb{H}_{\rm ideal}\left[x,\middle| y_{j}, a_{j}, j \leq k\right] \\
&~~~~+K\ln(1+\delta_K/\chin)\\
&~~~~+\sum_{j=K+1}^{s} \ln(1+(\delta_0+\delta_K)/\lambda_j). 
\end{align*}
The additional entropy relative to the ideal case $\mathbb{H}_{\rm ideal}$ is typically small, because $\delta_K \leq \delta_0 4^K$ (according to Lemma \ref{propdelta} in the appendix), $\delta_0$ is on the order of $\varepsilon^2$, and hence the second term in the appendix is on the order of $K^2$; the third term will be small because $\delta_0$ and $\delta_K$ are small compare to $\lambda_j$. 

\end{remark}

Note that, however, if the power allocations $\beta_i$ are calculated using the eigenvalues of the assumed covariance matrix $\widehat{\Sigma}$, after $K = s$ iterations, we are not guaranteed to reach the desired precision $\varepsilon$ with probability $p$. However, this becomes possible if we  increase the total power slightly. 
The following theorem establishes an upper bound on the amount of extra total power needed to reach the same precision $\varepsilon$ compared to the total power $P_{\rm ideal}$ if we use the correct covariance matrix.

\begin{theorem}[Additional power required]\label{thm:power}
Assume $K\leq s$ eigenvalues of $\Sigma$ are larger than $\chin$. 
If 
\[\|\widehat{\Sigma} - \Sigma\|\leq \frac{1}{4^{s+1}}\chin,\] 
then to reach a precision $\varepsilon$ at confidence level $p$, the total power $P_{\rm mismatch}$ required by Algorithm \ref{alg:Gaussian-all} when using $\widehat{\Sigma}$ is upper bounded by
\[P_{\rm mismatch} < P_{\rm ideal} + \left[\frac{20}{51}s+\frac{1}{272}K\right]\frac{\sigma^2}{\chin}.\]

\end{theorem}

\begin{remark}
In a special case when $K= s$ eigenvalues of $\Sigma$ are larger than $\chin$, 
 under the conditions of Theorem \ref{thm:power}, we have a simpler expression for the upper bound
\begin{align*}
P_{\rm mismatch}
& < P_{\rm ideal} + \frac{323}{816}  \frac{\sigma^2}{\chin}s.
\end{align*}
Note that the additional power required is quite small and is only linear in $s$. All other parameters are independent of the input matrix.
\end{remark}

\subsubsection{Initialize $\widehat{\Sigma}$}

In the following we present schemes to estimate $\widehat{\Sigma}$ to reach the desired precision in Theorem \ref{entropy}: (1) using sample covariance matrix if we are able to obtain full dimensional training samples; (2) using covariance sketching to estimate the covariance using random projections of the full dimensional training samples. 
%
%

Suppose the sample covariance matrix is obtained from training samples $\tilde{x}_1,\ldots,\tilde{x}_L$ that are drawn i.i.d. from  $\mathcal N(0,\Sigma)$, and $\widehat{\Sigma}=(1/L)\sum_{\ell=1}^{L}\tilde{x}_\ell \tilde{x}_\ell^{\intercal}.$ Then we need $L$ to be sufficiently large to reach the desired precision. The following Lemma \ref{cor:sampleSize} arises from a simple tail probability bound of the Wishart distribution (since the sample covariance matrix follows a Wishart distribution). 
\begin{lemma}[Initialize with sample covariance matrix]
\label{cor:sampleSize}
For any constant $\delta > 0$, we have $\Vert\widehat{\Sigma} -\Sigma\Vert \leq \delta$ with probability exceeding $1-2n\exp(-\sqrt{n})$, as long as
\[L\geq 4n^{1/2}{\rm tr}(\Sigma)(\Vert\Sigma\Vert/\delta^2+4/\delta).\]
\end{lemma}
Lemma \ref{cor:sampleSize} shows that the number of measurements needed  to reach a precision $\delta$ for a sample covariance matrix is $\mathcal{O}(1/\delta^2)$ as expected. 

We may also use a covariance sketching scheme similar to that described in \cite{CovSketching2012,CovSketchSparse,ChenChiGoldsmith2014} to estimate $\widehat{\Sigma}$. Covariance sketching is based on random projections of each training samples, and hence it is memory efficient when we are not able to store or operate on the full vectors directly. The covariance sketching scheme is described below and illustrated in Fig. \ref{cov_sketch}. Assume training samples $\tilde{x}_i$, $i = 1, \ldots, N$ are drawn from the signal distribution. Each sample, $\tilde{x}_i$ is sketched $M$ times using random sketching vectors $b_{ij}$, $j = 1, \ldots, M$, through a noisy linear measurement $(b_{ij}^\intercal x_i + w_{ijl})^2$, and we repeat this for $L$ times ($l = 1, \ldots, L$) and compute the average energy to suppress noise\footnote{Our sketching scheme is slightly different from that used in \cite{ChenChiGoldsmith2014} because we would like to use the square of the noisy linear measurements $y_i^2$ (where as the measurement scheme in \cite{ChenChiGoldsmith2014} has a slightly different noise model).  In practice, this means that we may use the same measurement scheme in the first stage as training to initialize the sample covariance matrix.}. This sketching process can be shown to be a linear operator $\mathcal{B}$ applied on the original covariance matrix $\Sigma$, as shown in Appendix \ref{app:cov_sketch}. 
We may recover the original covariance matrix from the vector of sketching outcomes $\gamma \in \mathbb{R}^M$ by solving the following convex optimization problem 
\begin{equation}\label{opt}
\begin{array}{rl}
\widehat{\Sigma}= \argmin_{X} & {\rm tr}(X)\\
{\rm subject\ to}& X\ \succeq 0,\ \Vert \gamma-\mathcal B(X)\Vert_1\leq \tau,
\end{array}
\end{equation}
where $\tau$ is a user parameter that depends on the noise level. In the following theorem, we further establish conditions on the covariance sketching parameters $N$, $M$, $L$, and $\tau$ so that the recovered covariance matrix $\widehat{\Sigma}$ may reach the required precision in Theorem \ref{entropy}, by adapting the results in \cite{ChenChiGoldsmith2014}.

\begin{figure}
\begin{center}
\includegraphics[width = 0.45\linewidth]{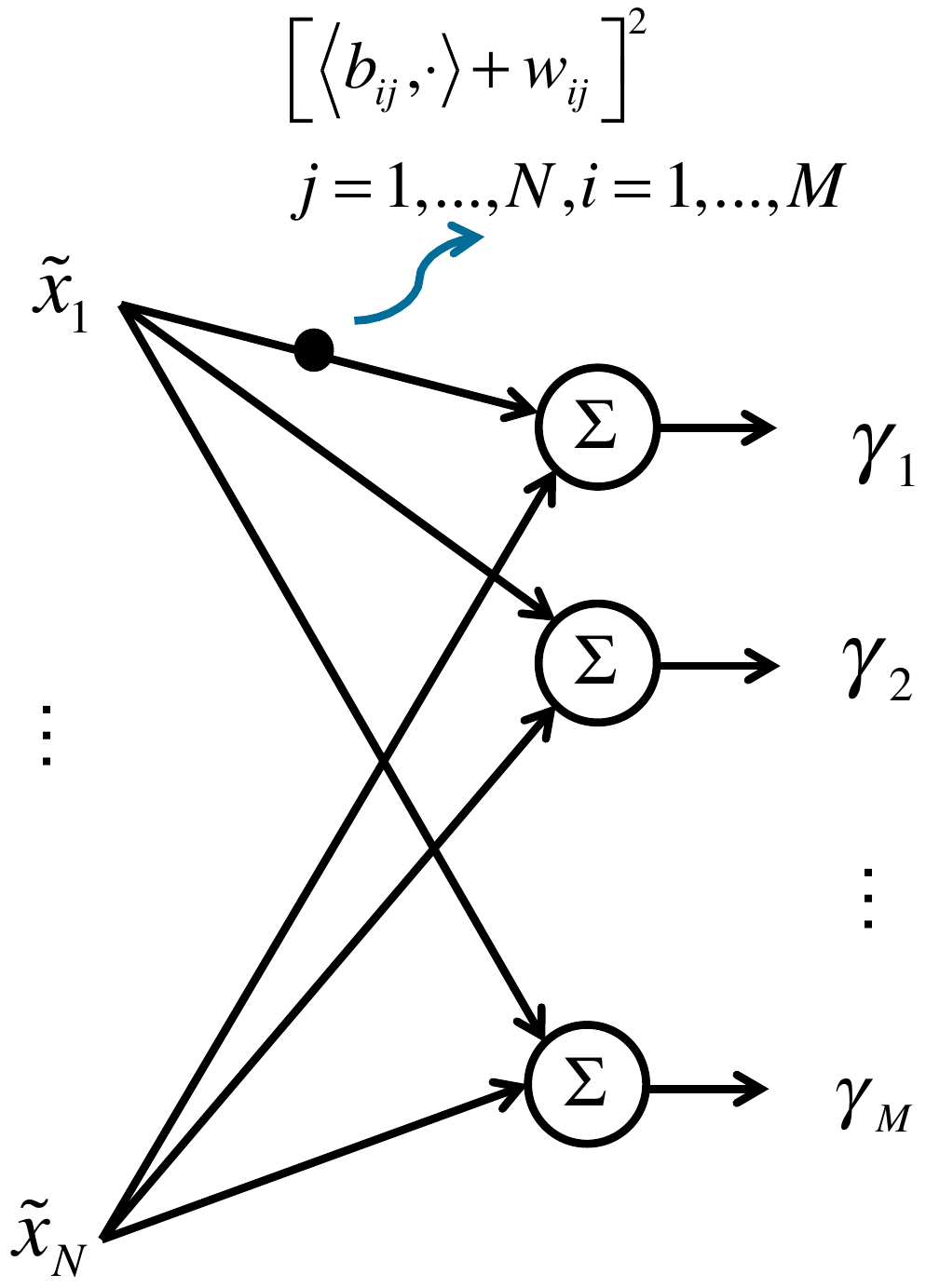}
\caption{Diagram of covariance sketching in our setting. The circle aggregates quadratic sketches from branches and computes the average. }
\label{cov_sketch}
\end{center}\vspace{-0.2in}
\end{figure}

\begin{lemma}[Initialize with covariance sketching]
\label{thm:cov-sketch}
For any $\delta > 0$ the solution to (\ref{opt}) satisfies 
$\Vert \widehat{\Sigma}-\Sigma \Vert\leq \delta,$
with probability exceeding
$1-2/n-{2}/{\sqrt{n}}-2n\exp(-\sqrt{n}) -\exp(-c_1 M)$,  as long as 
the parameters $M$, $N$, $L$ and $\tau$ satisfy the following conditions 
\begin{align}
M&>c_0 ns,\\
N 
&\geq 4n^{1/2}{\rm tr}(\Sigma)(\frac{36 M^2 n^2 \Vert\Sigma\Vert}{\tau^2}+\frac{24Mn}{\tau}),\\
L
&\geq \max\left\{ \frac{M}{4n^2\Vert\Sigma\Vert}\sigma^2, \ \frac{1}{\sqrt{2[{\rm tr}{(\Sigma)}/\Vert\Sigma\Vert] Mn^2}}\sigma^2, \frac{6M}{\tau}\sigma^2\right\},\label{L_LB}\\
\tau& =M \delta/c_2, \label{def_tau}
\end{align}
where $c_0$, $c_1$, and $c_2$ are absolute constants. 
\end{lemma}

\subsection{Gaussian mixture model (GMM)}

We also establish a lower bound on the number of measurements (or power) required to recover a GMM signal with high precision when there is model mismatch. The proof follows by identifying a connection between the Info-Greedy Sensing and the so-called \emph{multiplicative weight update} (MWU) algorithms (see e.g.,
\cite{cesa2006prediction,nisan2007algorithmic,arora2012multiplicative}). The
MWU method is actually a meta-algorithm and its instantiations span a large family of
algorithms. It has been re-derived under various names in various disciplines. MWU algorithms maintain a distribution over experts (which corresponds to different Gaussian components in our case) and form a solution by e.g., a majority vote or an average over the solutions suggested by the experts. (Here each Gaussian component will suggest a sensing vector.) The weights are updated in each round according to the posterior update. We will use the hedge version of MWU in deriving the result. 

\begin{theorem}[GMM with Mismatch]
\label{lemma:GMM_mismatch}
Denote the posterior mean and covariance of component $c$ after $k$ iterations as $\mu_{c,k}$ and $\Sigma_{c,k}$, and their perturbed counterparts as $\hat{\mu}_{c,k}$ and $\widehat{\Sigma}_{c,k}$, respectively. Let $\delta_{c,k}=\Vert \widehat{\Sigma}_{c,k}-\Sigma_{c,k} \Vert$, 
and $m_c$ be the number of measurements (or power) required to ensure $\|x-\hat{x}\| <\varepsilon$ with probability $p$ for a Gaussian signal $\mathcal{N}(\mu_c, \Sigma_c)$ corresponding to component $c$ for all $c \in [C]$ if we start with sample covariance matrix $\widehat{\Sigma}_c$.
%
Then if both the mismatch in the initial mean $|a^\intercal(\mu - \hat{\mu})|$ and the initial covariance $\|\widehat{\Sigma}-\Sigma\|$ are sufficiently small so that $\eta_0 = \mathcal{O}(\hat{\eta})$, then
for a signal sampled from the $c^*$th component, we need at most 
\[
 \sum_{c=1}^C m_{c}+\mathcal{O}\left(\frac{{\rm ln}\vert C \vert}{\hat{\eta}+\eta_0}\right)
\]
amount of power to ensure $\|x - \widehat{x}\| < \varepsilon$ with probability $p(1-\hat{\eta} - 1/n)$ when sampling from the posterior distribution of $\pi$ with probability. Here $\hat{\eta} = 1/2$, 
\[
\eta_0 = \frac{1}{\sigma^4}\cdot \max_{k=1}^n \left \{\sigma^2(\vert\varrho_k\vert+2n{b_k})\vert\varrho_k\vert +\beta_k n^2 {b_k^2}\delta_{k-1}\right\},
\]
\[
\varrho_k \triangleq a_k^\intercal (\hat{\mu}_{c^*, k}-\mu_{c^*, k}),
\]
\[
b_k = \sqrt{(\lambda_{c^*,k} + \delta_{c^*,k-1}) + \sigma^2 + (a_k^\intercal(\mu_{c^*} - \mu_{c^*,k-1}))^2}. 
\]
\end{theorem}
Note that here the constants are defined in terms of maximizing from $1$ to $n$. This can be understood as if we run the algorithm until we have aquired $n$ measurements.
\begin{remark}
If our goal is to detect the correct component (rather than recovering the signal itself), we need at most $\mathcal{O}\left(\frac{{\rm ln}\vert C \vert}{\hat{\eta}+\eta_0}\right)$ samples if the true component is $c^*$.
\end{remark}

\begin{remark}
Compared to GMM result without mismatch, which is on the order of $\mathcal{O}(\ln|C|/\hat{\eta})$ \cite{InfoGreedy2014}, this upper bound actually requires a smaller number of measurements $\mathcal{O}(\ln|C|/(\hat{\eta}+\eta_0))$. This is consistent with our intuition, and it says that if our estimation accuracy for the covariance matrices is low, then we should not ``labor'' as much. Because the estimation error will create an ``error floor'' which does not decrease by making more measurements, and it is not meaningful to make additional measurements below the noise floor. Of course, when there is covariance error, the overall error bound will be higher as well (which is already captured by a larger error bound).
\end{remark}

\subsection{One-sparse measurement}

In the following we provide performance bounds for the case of one-sparse measurements  in Algorithm \ref{alg:one-sparse}. Assume the signal covariance matrix is known precisely. 
Now that $\Vert a_k \Vert_0=1$, we have $a_k=\sqrt{\beta_k} u_k$, where $u_k\in \{e_1,\cdots, e_n\}$.
Suppose the largest diagonal entry of $\Sigma^{(k-1)}$ is determined by
\[
j_{k-1}={\rm arg} \max_{t}\Sigma_{tt}^{(k-1)}.
\]
From the update equation for the covariance matrix in Algorithm \ref{alg:one-sparse}, 
the largest diagonal entry of $\Sigma^{(k)}$ can be determined from  
\[
j_{k}={\rm arg}\max_{t}~\left\{ \Sigma_{tt}^{(k-1)}-\frac{(\Sigma_{t j_{k-1}}^{(k-1)})^2}{\Sigma_{j_{k-1}j_{k-1}}^{(k-1)}+\sigma^2/\beta_k}\right\}.
\]
Let the correlation coefficient be denoted as \[\rho_{ij}^{(k)}=\frac{(\Sigma_{ij}^{(k)})^2}{\Sigma_{ii}^{(k)}\Sigma_{jj}^{(k)}},\]  where the covariance of the $i$th and $j$th coordinate of $x$ after $k$ measurements is denoted  as $\Sigma_{ij}^{(k)}$. 
\begin{lemma}[One sparse measurement. Recursion for trace of covariance matrix]\label{lemma:one_sparse}
Assume the minimum correlation for the $k$th iteration is $\rho^{(k-1)} \in [0, 1)$ such that $\rho^{(k-1)}\leq |\rho_{ij_{k-1}}^{(k-1)}|$ for any $i\in [n]$. Then for a constant $\gamma>0$, if the power of the $k$th measurement $\beta_k$ satisfies $\beta_k\geq {\sigma^2}/\left({\gamma\max_{t}\Sigma_{tt}^{(k-1)}}\right)$, we have
\begin{equation}
{\rm tr}(\Sigma_k)\leq \left[1-\frac{(n-1)\rho^{(k-1)}+1}{n(1+\gamma)}\right]{\rm tr}(\Sigma_{k-1}).
\label{trace_bound}
\end{equation}
\end{lemma}

Lemma \ref{lemma:one_sparse} provides a good bound for a one-step ahead prediction for the trace of the covariance matrix, as demonstrated in Fig. \ref{fig:bound}. 
Using the above lemma, we can obtain an upper bound on the number of measurements needed for one-sparse measurements. 
\begin{theorem}[Gaussian, one-sparse measurement]\label{Gaussian_one_sparse}
For constant $\gamma>0$, when power is allocated satisfying $\beta_k\geq {\sigma^2}/({\gamma\max_{t}\Sigma_{tt}^{(k-1)}})$ for $k=1,2,\ldots, K$, we have $\|\hat{x} - x\|\leq \varepsilon$ with probability $p$ as long as
\begin{equation}
K\geq \max\Big\{\frac{\ln[{\rm tr}(\Sigma)/\chin]}{\ln\frac{1}{1-{1}/[n(1+\gamma)]}},0\Big \}.
\label{K_one_sparse}
\end{equation}
\end{theorem}

\begin{figure}[h!]
\begin{center}
\includegraphics[width = 0.7\linewidth]{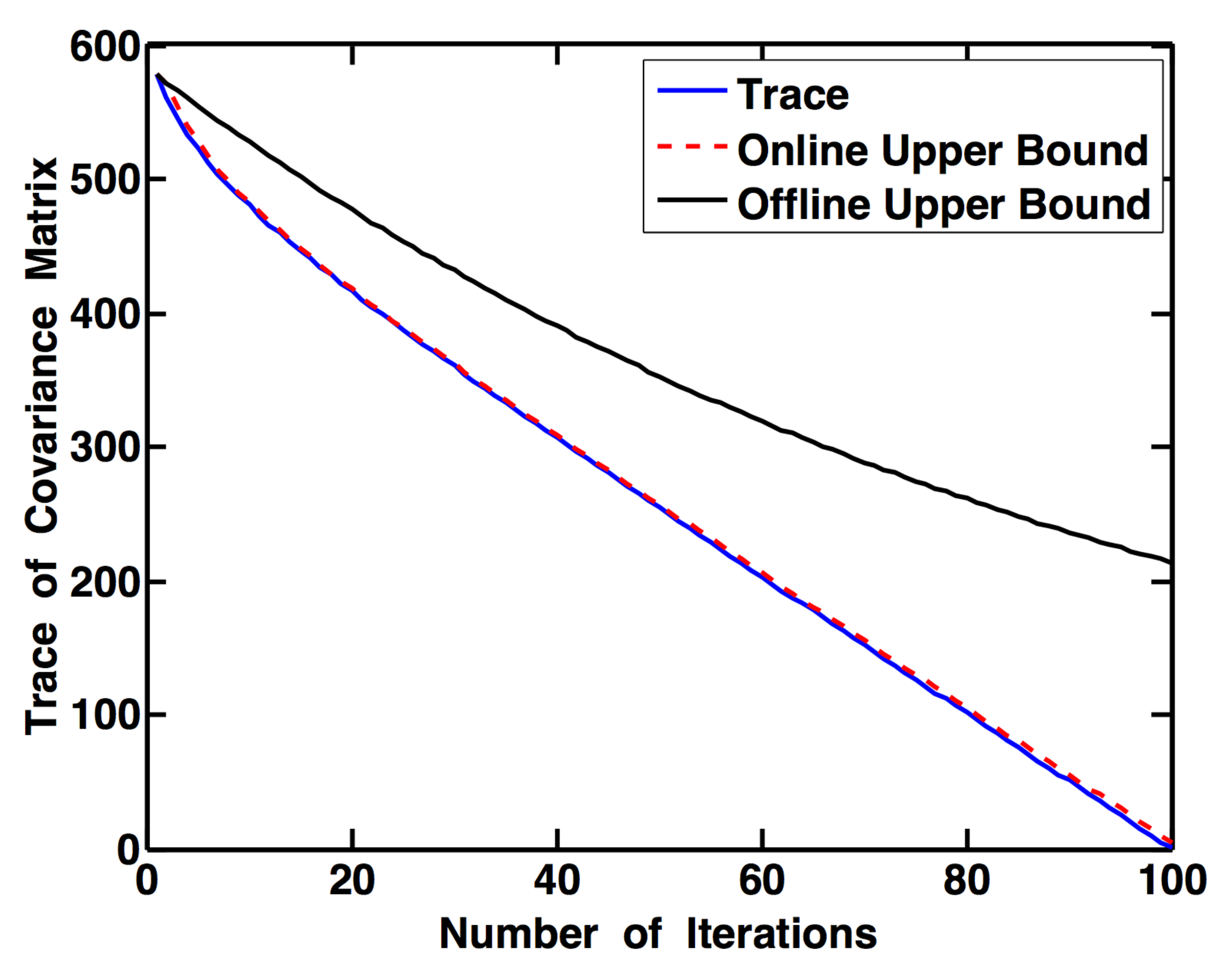}
\end{center}
\caption{One-step ahead prediction for the trace of the covariance matrix: the offline bound corresponds to applying (\ref{trace_bound}) iteratively $k$ times, and the online bound corresponds to predicting ${\rm tr}(\Sigma_k)$ using ${\rm tr}(\Sigma_{k-1})$. Here $n=100$, $p=0.95$, $\varepsilon=0.1$, $\Sigma=dd^{\intercal}+5I_n$ where $d=[1,\cdots,1]^{\intercal}$. 
}
\label{fig:bound}
\end{figure}

The above theorem requires the number of iterations to be on the order of $\ln(1/\varepsilon)$ to reach precision $\varepsilon$ (recall $\chin = \varepsilon^2/\chi_n^2(p)$), as expected. It also suggest a method to allocate power: set $\beta_k$ to be proportional to $\sigma^2/\max_{t}\Sigma_{tt}^{(k-1)}$: this captures the inter-dependence of the signal entries as these dependence will be affect the diagonal entries of the updated covariance matrix. 

\section{Numerical examples}\label{sec:num_egs}

In the following, we have three sets of numerical examples to demonstrate the performance of Info-Greedy Sensing when there is mismatch in the signal covariance matrix, when the signal is sampled from Gaussian, and from GMM models, respectively.

\subsection{Sensing Gaussian with mismatched covariance matrix}

When the assumed covariance matrix for the signal $x$ is equal to its true covariance matrix, Info-Greedy Sensing is identical to the batch method \cite{CarsonChenRodrigues2012} (the batch method measures using the largest eigenvectors of the signal covariance matrix). However, when there is a mismatch between the two, Info-Greedy Sensing outperforms the batch method due to its adaptivity, as shown by the example demonstrated in Fig. \ref{Fig:mismatch} (with $K = 20$). Further performance improvement can be achieved by updating the covariance matrix using estimated signal sequentially such as described in (\ref{seq_update_cov}). Info-Greedy Sensing also outperforms the sensing algorithm where $a_i$ are chosen to be random Gaussian vectors with the same power allocation, as it uses prior knowledge (albeit being imprecise) about the signal distribution. 

Fig. \ref{Fig:full_unit} demonstrates an effect that when there is a mismatch in the assumed covariance matrix, better performance can be achieved if we make many lower power measurements than making one full power measurement because we update the assumed covariance matrix in between.

\begin{figure}[h!]
\begin{center}
\begin{tabular}{c}
\includegraphics[width = 0.7\linewidth]{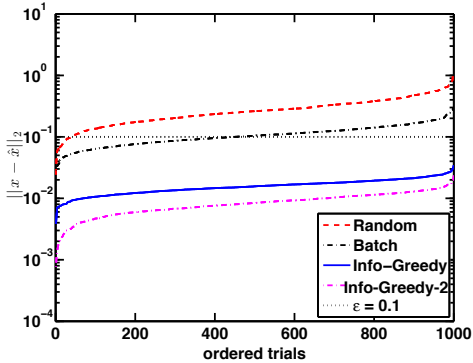}\\
\end{tabular}
\end{center}
\caption{Sensing a Gaussian signal of dimension $n= 100$, when there is mismatch between the assume covariance matrix and the true covariance matrix: $\widehat{\Sigma} \propto \Sigma + RR^\intercal$, where $R\in \mathbb{R}^{n\times 3}$  and each entry of $R_{ij} \sim \mathcal{N}(0, 1)$. We repeat 1000 Monte Carlo trials and for each trial we use $K = 20$ measurements. The Info-Greedy-2 method corresponds to (\ref{seq_update_cov}), where we update the assumed covariance matrix sequentially each time we recover a signal and $\alpha = 0.5$. }
\label{Fig:mismatch}
\end{figure}

\begin{figure}
\begin{center}
\includegraphics[width = 0.7\linewidth]{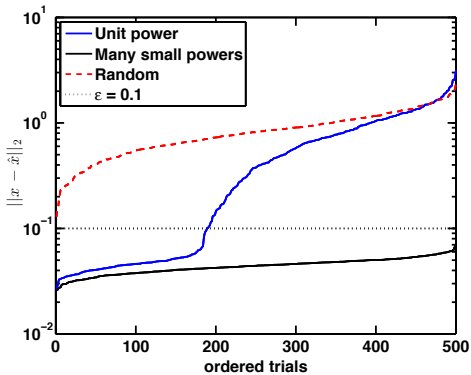}
\end{center}
\caption{Comparison of sensing a Gaussian signal with dimension $n = 100$ using unit power measurements along the eigenvector direction, versus splitting each unit-power measurement into 5 smaller ones, each with amplitude $\sqrt{1/5}$, and we update the covariance matrix in between. The mismatched covariance matrix is $\widehat{\Sigma} \propto \Sigma + rr^\intercal$, where $r\in \mathbb{R}^{n\times 5}$ and each entry of $r$ is i.i.d. $\mathcal{N}(0, 1)$, and $\widehat{\Sigma}$ is normalized to have unit spectral norm.
}
\label{Fig:full_unit}
\end{figure}

\subsection{One-sparse measurements}

In this example, we sense a GMM signal with a one-sparse measurement. Assume there are $C = 3$ components and we know the signal covariance matrix exactly. We consider two cases of generating the covariance matrix for each signal: when the low rank covariance matrices for each component are generated completely at random, and when it has certain structure. Fig. \ref{fig:GMM-one-sparse-recon} shows the reconstruction error $\|\hat{x} - x\|$, using $K = 40$ one-sparse measurements for GMM signals. Note that Info-Greedy Sensing (Algorithm \ref{alg:one-sparse}) with unit power $\beta_j = 1$ can significantly outperform the random approach with unit power (which corresponds to randomly selecting coordinates of the signal to measure). Fig. \ref{fig:one-sparse-GMM} also compares the mis-classification rate of Info-Greedy Sensing with one-sparse measurements to that with using a full signal vector $x$ for classification. Note that, interestingly, using $K = 50$ one-sparse measurements we can obtain a performance very similar to the ideal case,  which can be explained since we exploit the correlation structure of the signal. 



\begin{figure}[h!]
\begin{center}
\begin{tabular}{c}
\includegraphics[width = 0.7\linewidth]{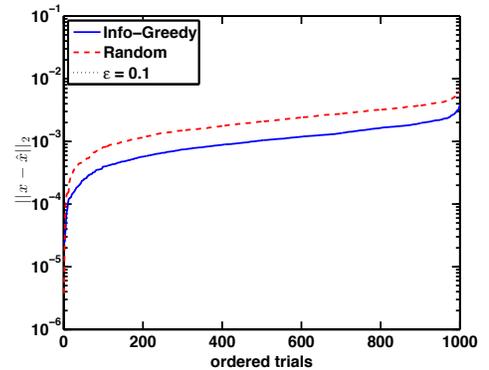}\\
(a)  $\Sigma_c \propto RR^{\intercal}$, $R\in \mathbb{R}^{n\times 3}$, $R_{ij} \sim \mathcal{N}(0, 1)$\\
\includegraphics[width = 0.7\linewidth]{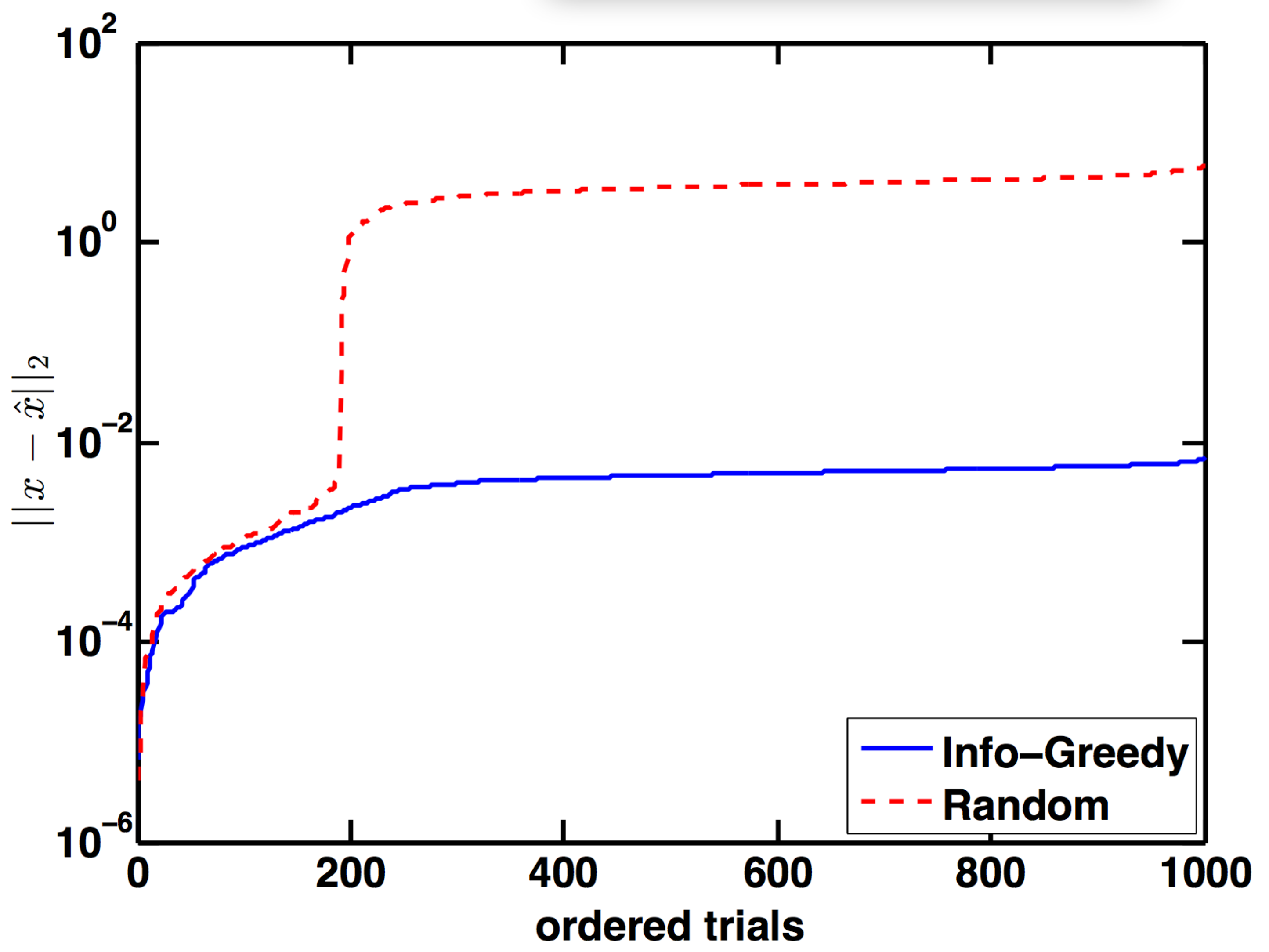}\\
(b) $\Sigma_c \propto (11^\intercal + 20 \alpha^2 \cdot\mbox{diag}\{n,n-1,\cdots,1\})$, $\alpha \sim \mathcal{N}(0, 1)$
\end{tabular}
\end{center}
\caption{Sensing a low-rank GMM signal of dimension $n=100$ using $K=40$ measurements with $\sigma = 0.001$, when the covariance matrices are generated (a) completely randomly, or (b) having certain structure. The covariance matrices $\Sigma_c$ are normalized so that their spectral norms are 1.  }
\label{fig:GMM-one-sparse-recon}
\end{figure}


\begin{figure}
\begin{center}
\includegraphics[width = 0.7\linewidth]{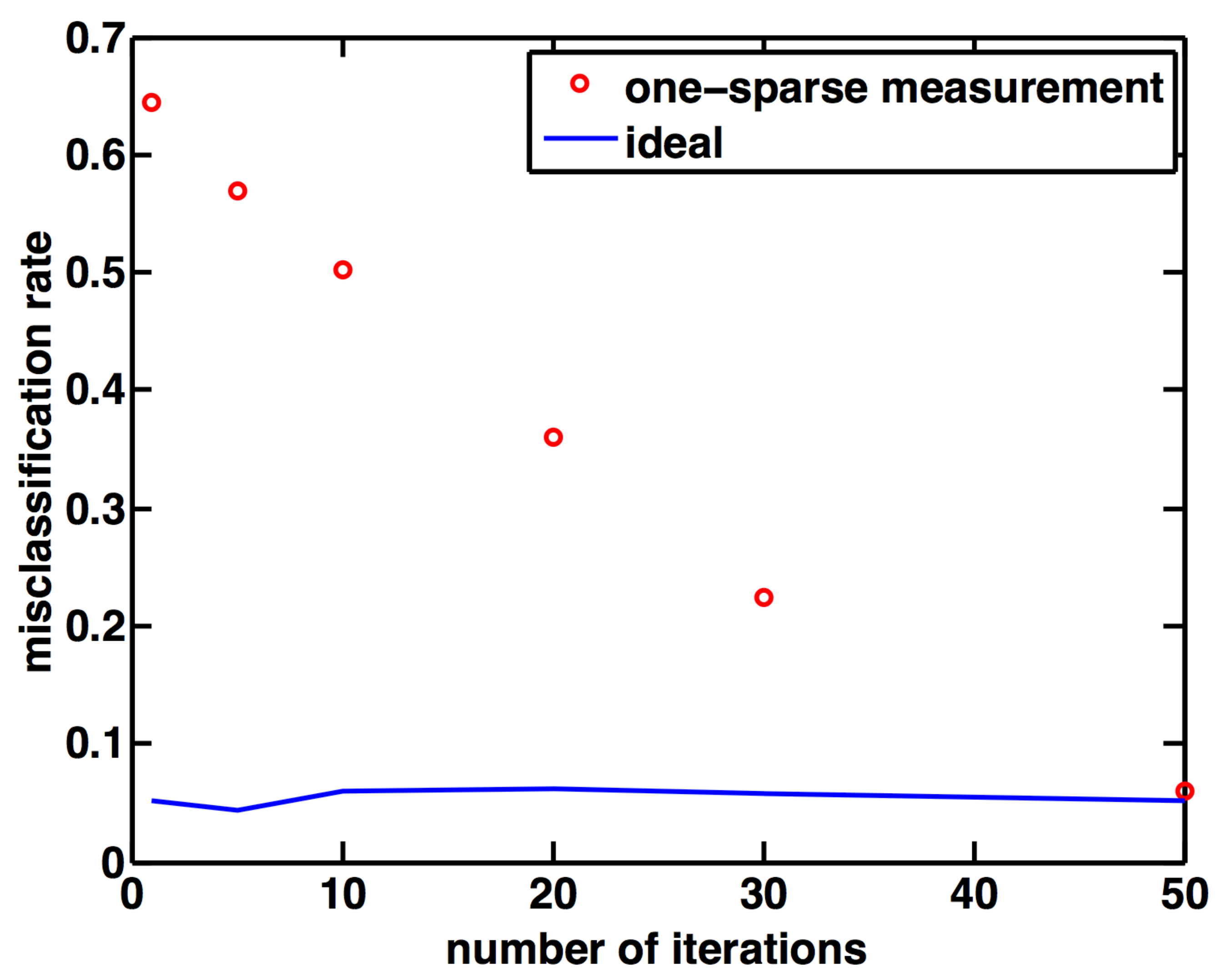}
\end{center}
\caption{Classifying a signal of dimension $n = 100$ generated from a GMM model with covariance matrix generated according to $\Sigma \propto RR^{\intercal}$ and the true distribution is $\pi=(0.5, 0.3, 0.2)$. We assume a uniform initial distribution $(1/3, 1/3, 1/3)$. Misclassification rate versus the number of measurements $K$. Ideal case corresponds to where we observe $x$ and run a quadratic discriminate analysis using the full vector $x$ (i.e. rather than just observing a noisy version of an entry of $x$ each time). }
\label{fig:one-sparse-GMM}
\end{figure}

\subsection{Real data}
\subsubsection{Sensing of a video stream using Gaussian model} 

In this example, we use a video from the Solar Data Observatory. The frame is of size $232\times 292$ pixels. We use the first 50 frames to form a sample covariance matrix $\widehat{\Sigma}$, and use it to perform Info-Greedy Sensing on the rest of the frames. We take $K=90$ measurements. As demonstrated in Fig. \ref{solar-flare}, Info-Greedy Sensing performs much better in that it acquires more information such that the recovered image has much richer details.
\begin{figure}[h!]
\begin{center}
\begin{tabular}{cc}
\includegraphics[width = 0.45\linewidth]{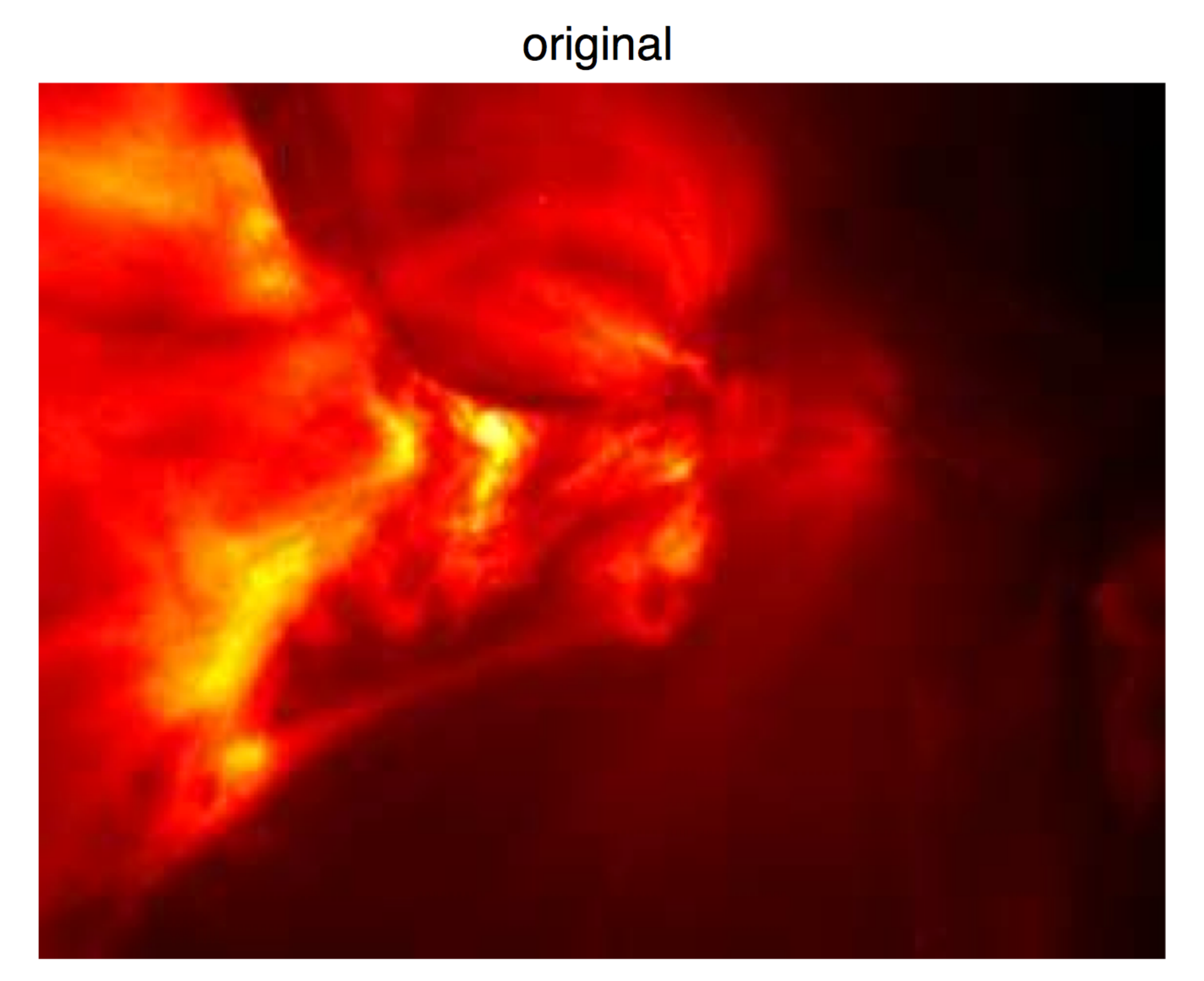}  &
\includegraphics[width = 0.45\linewidth]{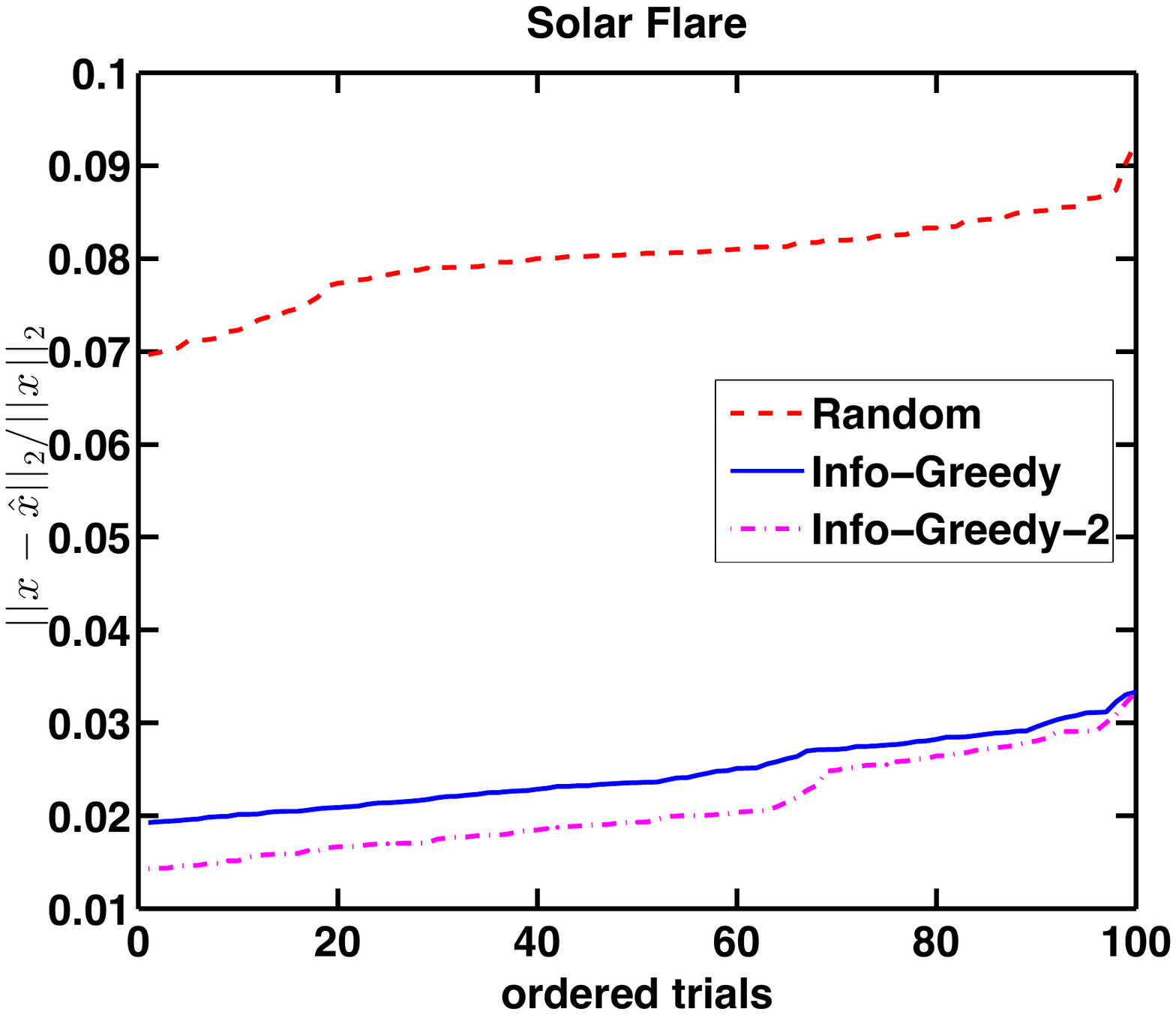}\\
(a) & (b) \\
\includegraphics[width = 0.45\linewidth]{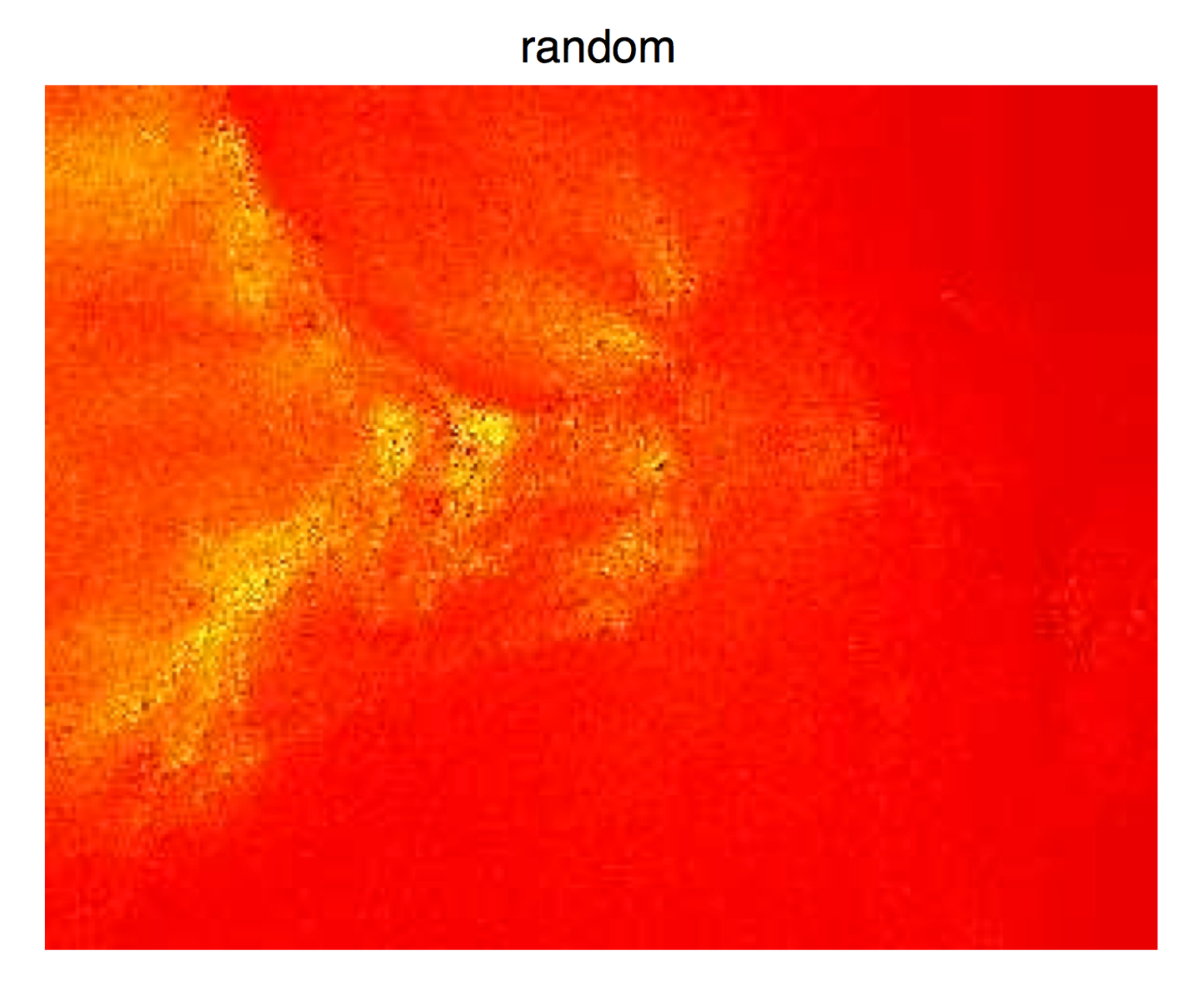}
&
\includegraphics[width = 0.45\linewidth]{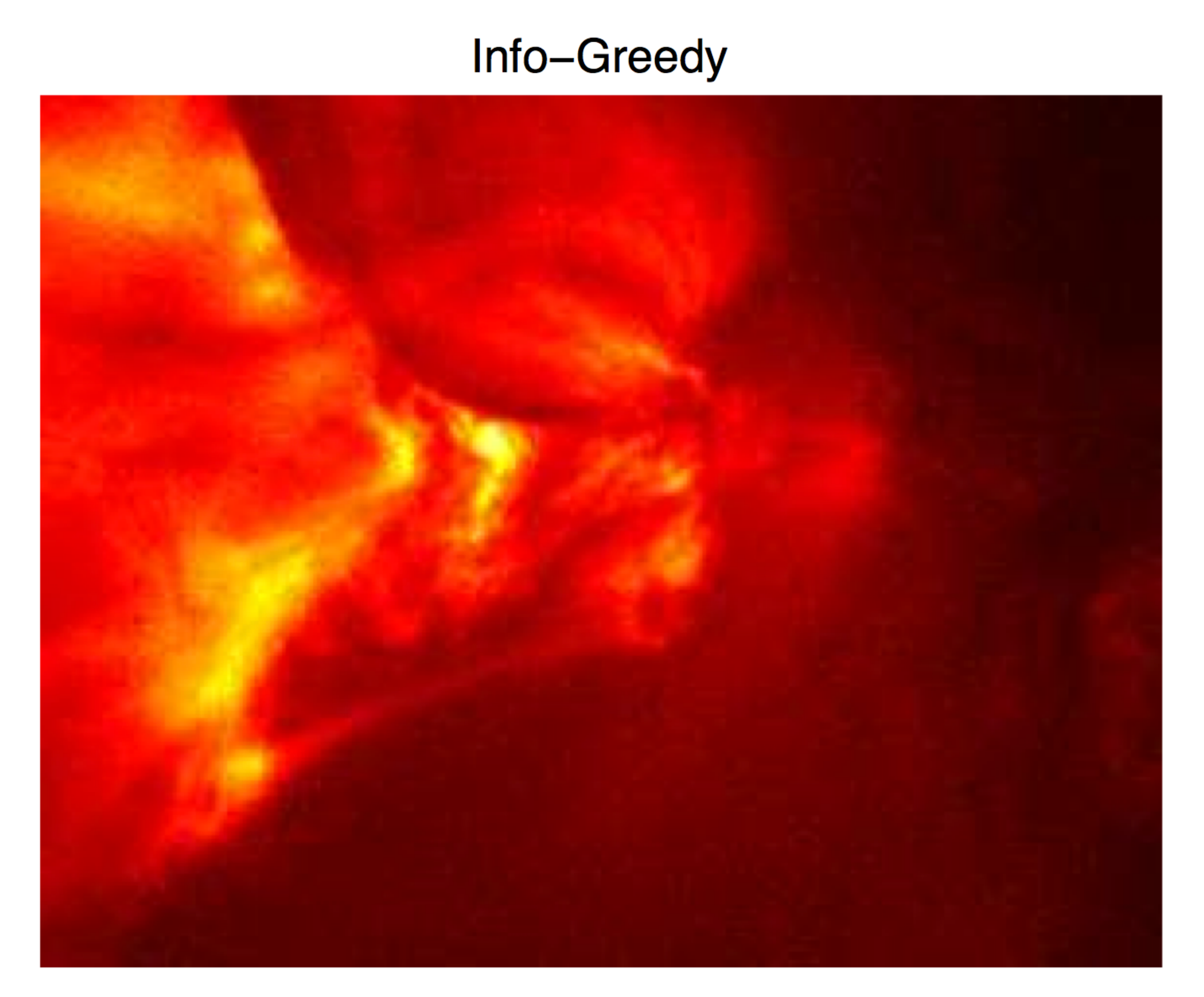}
\\
(c) & (d)
\end{tabular}
\end{center}
\caption{Recovery of solar flare images of size 224 by 288 with $K = 90$ measurements and no sensing noise. We used the first 50 frames to estimate the mean and covariance matrix of a single Gaussian. (a) original image for 300th frame; (b) ordered relative recovery error of the 200th to the 300th frames; (c) recovered the 300th frame using random measurement; (d) recovered the 300th frame using Info-Greedy Sensing.}
\label{solar-flare}
\end{figure}

\subsubsection{Sensing of a high-resolution image using GMM}\label{GT_img}
 We consider a scheme for sensing a high-resolution image that exploits the fact that the patches of the image can be approximated using a Gaussian mixture model, as demonstrated in Fig. \ref{high_rel_gt}. 
%
We break the image into 8 by 8 patches, which result in 89250 patches. We randomly select 500 patches (0.56\% of the total pixels) to estimate a GMM model with $C = 10$ components, and then based on the estimated GMM initialize Info-Greedy Sensing with $K = 5$ measurements and sense the rest of the patches. This means we can use a compressive imaging system to capture 5 low resolution images of size 238-by-275 
(this corresponds to compressing the data into $8.32\%$ of its original dimensionality). With such a small number of measurements, the recovered image from Info-Greedy Sensing measurements has superior quality compared with those with random masks. 


\section{Conclusions and discussions}

In this paper, we have explored the value of information and how to use such information in sequential compressive sensing, by examining the Info-Greedy Sensing algorithms when the signal covariance matrix is not known exactly. We quantify the algorithm performances in the presence of estimation errors and when only one-sparse measurements are allowed. 

Our results for Gaussian and GMM signals are quite general in the following sense. 
In high-dimensional problems, a commonly used low-dimensional signal model
for $x$ is to assume the signal lies in a
 subspace plus Gaussian noise, which corresponds to the case where the signal is Gaussian with a low-rank covariance matrix; GMM is also commonly used (e.g., in image analysis and video processing) as it models signals lying in a union of multiple subspaces plus Gaussian noise. 
In fact,
parameterizing via low-rank GMMs is a popular way to approximate
complex densities for high-dimensional data. Hence, we may be able to
couple the results for Info-Greedy Sensing of GMM with the
recently developed methods of scalable multi-scale density estimation
based on empirical Bayes \cite{WangCanale2014} to create powerful
tools for information guided sensing for a general signal model. We
may also be able to obtain performance guarantees using multiplicative weight update techniques
together with the error bounds in \cite{WangCanale2014}.

%

%

\bibliographystyle{ieeetr}
\bibliography{bib}


\appendices

\section{Covariance sketching}\label{app:cov_sketch}

Consider the following setup for covariance sketching. Suppose we are able to form measurement in the form of $y = \transpose a x + w$ like we have in the Info-Greedy Sensing algorithm. 
Suppose there are $N$ copies of Gaussian signal we would like to sketch: $\tilde{x}_1,\ldots, \tilde{x}_N$ that are i.i.d. sampled from $\mathcal{N}(0, \Sigma)$, and we sketch using $M$ random vectors: $b_1, \ldots, b_M$. Then for each fixed sketching vector $b_i$, and fixed copy of the signal $\tilde{x}_j$, we acquire $L$ noisy realizations of the projection result $y_{ijl}$ via
\[y_{ijl}=\transpose b_i \tilde{x}_j +w_{ijl}, \quad l = 1, \ldots, L.\]
We choose the random sampling vectors $b_i$ as i.i.d. Gaussian with zero mean and covariance matrix equal to an identity matrix. Then we average $y_{ijl}$ over all realizations $l = 1, \ldots, L$ to form the $i$th sketch $y_{ij}$ for a single copy $\tilde{x}_j$: 
\[y_{ij}=\transpose b_i\tilde{x}_j +\underbrace{\frac{1}{L}\sum_{l=1}^L w_{ijl}}_{w_{ij}}.\]
The average is introduced to suppress measurement noise, which can be viewed as a generalization of sketching using just one sample. 
Denote $w_{ij}=\frac{1}{L}\sum_{l=1}^L w_{ijl}$, which is distributed as $\mathcal N(0, \sigma^2/L)$. Then we will use the average energy of the sketches as our data $\gamma_i$, $i = 1, \ldots, M$, for covariance recovery:
\[
\gamma_i \triangleq \frac{1}{N}\sum_{j=1}^{N}y_{ij}^2.
\]
Note that $\gamma_i$ can be further expanded as 
\begin{equation}
\gamma_i = {\rm tr}(\widehat{\Sigma}_N b_i \transpose b_i)+\frac{2}{N}\sum_{j=1}^N w_{ij}\transpose b_i \tilde{x}_j +\frac{1}{N}\sum_{j=1}^N w_{ij}^2, \label{meas}
\end{equation}
where \[\widehat{\Sigma}_N=\frac{1}{N}\sum_{j=1}^{N}\tilde{x}_j \tilde{x}^{\intercal}_j,\] is the maximum likelihood estimate of $\Sigma$ (and is also unbiased). We can write (\ref{meas}) in vector matrix notation as follows. 
Let
$\gamma=[\gamma_1,\cdots \gamma_M\transpose]$. Define a linear operator 
 $\mathcal B:\mathbb R^{n\times n}\mapsto \mathbb R^{M}$ such that  $\mathcal [B(X)]_i={\rm tr}(X b_i \transpose b_i)$. Thus, we can write (\ref{meas}) as a linear measurement of the true covariance matrix $\Sigma$  
 \[\gamma=\mathcal{B} (\Sigma)+\eta,\]
where $\eta \in \mathbb{R}^M$ contains all the error terms and corresponds to the noise in our covariance sketching measurements, with the $i$th entry given by $$\eta_i=\transpose b_i(\widehat{\Sigma}_N-\Sigma) b_i+\frac{2}{N}\sum_{j=1}^N w_{ij}\transpose b_i \tilde{x}_j +\frac{1}{N}\sum_{j=1}^N w_{ij}^2.$$
Note that we can further bound the $\ell_1$ norm of the error term as
\begin{align*}
\Vert\eta\Vert_1 =\sum_{i=1}^M\vert\eta_i\vert
\leq \Vert\widehat{\Sigma}_N-\Sigma\Vert b+ 2\sum_{i=1}^M\vert z_i\vert+w,
\end{align*}
where $$b=\sum_{i=1}^M \Vert b_i\Vert^2,\ \mathbb E[b]=Mn,\ {\rm Var}[b] =2Mn,$$
$$w=\frac{1}{N}\sum_{i=1}^M \sum_{j=1}^N w_{ij}^2,\ \mathbb E[w]=M\sigma^2/L,\ {\rm and}\ {\rm Var} [w]=\frac{2M\sigma^4}{NL^2},$$
$$z_i=\frac{1}{N}\sum_{j=1}^N w_{ij}\transpose b_i \tilde{x}_j,\ \mathbb E[z_i]=0\ {\rm and}\ {\rm Var} [z_i]=\frac{\sigma^2 {\rm tr}(\Sigma)} {NL}.$$
We may recover the true covariance matrix from the sketches $\gamma$ using the convex optimization problem (\ref{opt}).

\section{Backgrounds}

\begin{lemma}[Eigenvalue of perturbed matrix \cite{stewart1990matrix}]\label{lemmaeig}
Let $\Sigma$, $\widehat{\Sigma}\in\mathbb{R}^{n\times n}$ be symmetric,with eigenvalues $\lambda_1\geq\cdots\geq \lambda_n$ and $\hat{\lambda}_1\geq\cdots\geq\hat{\lambda}_n$, respectively. Let $E=\widehat{\Sigma}-\Sigma$ have eigenvalues $e_1\geq\cdots\geq e_n$. Then for each $i\in\{1,\cdots, n\}$, the perturbed eigenvalues satisfy 
$\hat{\lambda}_i\in[\lambda_i+e_n, \lambda_i+e_1].$
\end{lemma}


\begin{lemma}[Stability conditions for covariance sketching\cite{ChenChiGoldsmith2014}]
\label{sketching}
Denote $\mathcal A:\mathbb R^{n\times n}\mapsto \mathbb R^{m}$ a linear operator and for $X\in \mathbb R^{n\times n}$, $\mathcal A(X)=\{a_i^T X a_i\}_{i=1}^{m}$. Suppose the measurement is contaminated by noise $\eta\in R^m$, i.e., $Y=\mathcal A(\Sigma)+\eta$ and assume $\Vert \eta \Vert_1\leq \epsilon_1$. Then with probability exceeding $1-\exp (-c_1 m)$ the solution $\widehat{\Sigma}$ to the trace minimization (\ref{opt}) satisfies
\[\Vert \widehat{\Sigma}-\Sigma \Vert_F\leq c_0 \frac{\Vert\Sigma-\Sigma_r\Vert_{*}}{\sqrt{r}}+c_2\frac{\epsilon_1}{m},\]
for all $\Sigma \in R^{n\times n}$, provided that $m> c_0nr$. Here $c_0$, $c_1$, and $c_2$ are absolute constants and $\Sigma_r$ represents the best rank-r approximation of $\Sigma$.
When $\Sigma_r$ is exactly rank-$r$
$$\Vert \widehat{\Sigma}-\Sigma \Vert_F\leq c_0\frac{\epsilon_1}{m}.$$
\end{lemma}
\begin{lemma}[Concentration of measure for Wishart distribution \cite{zhu2012short}]\label{zhu2012}
\label{Tail}
If $X\in \mathbb R^{n\times n}\sim \mathcal W_n(N,\Sigma)$, then for $t>0$,
$$P\{\Vert\frac{1}{N}X-\Sigma\Vert\geq (\sqrt{\frac{2t(\theta+1)}{N}}+\frac{2t\theta}{N})\Vert\Sigma \Vert\}\leq 2n\exp(-t),$$
where $\theta=\rm tr(\Sigma)/\Vert \Sigma \Vert.$
\end{lemma}

\clearpage


\section{Proofs}

\subsection{Gaussian signal, with mismatch}

\begin{proof}[Proof of Theorem \ref{lemma_bias}]
Let $\xi_k=\hat{\mu}_k-\mu_k. $
From the update equation for the mean
 \(\hat{\mu}_k = \hat{\mu}_{k-1}
        + \widehat{\Sigma}_{k-1} a_k
        (y_k - \transpose a_k \hat{\mu}_{k-1})/(\hat{a}_k^\intercal \widehat{\Sigma}_{k-1} a_k +\sigma^2),\)  since $a_{k}$ is eigenvector of $\hat{\Sigma}_{k-1}$, 
we have the following recursion:
\begin{equation}
\begin{split}
\xi_k&=(I_n-\frac{\hat{\lambda}_{k}a_ka_k^{\intercal}}{\beta_k\hat{\lambda}_k+\sigma^2})\xi_{k-1} \\&+
\left[-\hat{\lambda}_k    \frac{a_k^{\intercal}E_{k-1}a_k}{(\beta_k\hat{\lambda}_{k}+\sigma^2-a_k^{\intercal}E_{k-1}a_k){(\beta_k\hat{\lambda}_{k}+\sigma^2})}a_k\right.\\
&~~~~~\left.+\frac{E_{k-1}a_k}{\beta_k\hat{\lambda}_{k}+\sigma^2-a_k^{\intercal}E_{k-1}a_k}\right](a_k^\intercal (x-\mu_{k-1}) + w_k).
\end{split}
 \label{recursion_xi}
\end{equation}
From the recursion of $\xi_k$ in (\ref{recursion_xi}), for some vector $C_k$ defined properly, we have that
\begin{equation}
\begin{split}
\mathbb E[\xi_k]=&~(I-\frac{\hat{\lambda}_k \beta_k}{\beta_k\hat{\lambda}_k+\sigma^2}u_k u_k^{\intercal}) \mathbb E[\xi_{k-1}]\\
&+C_k \underbrace{\mathbb E[a_k^{\intercal}(x-\mu_{k-1})+w_k]}_{0}.
\end{split}\label{recursion_xi_2}
\end{equation}
Note that the second term is equal to zero using an argument based on iterated expectation 
\[
\mathbb E[a_k^{\intercal}(x-\mu_{k-1}) + w_k]=a_k^\intercal \mathbb E[\mathbb E[x-\mu_{k-1}|y_1, \ldots, y_k]]=0.
\]
Hence Theorem \ref{lemma_bias} is proved by iteratively apply the recursion (\ref{recursion_xi_2}). When $\hat{\mu}_0 - \mu_0=0$, we have 
$
\mathbb E[\xi_k]=0, k=0,1,\ldots, K.
$
\end{proof}

In the following, Lemma \ref{propdelta} to Lemma \ref{proprank} are used to prove Theorem \ref{entropy}.

\begin{lemma}[Recursion in covariance matrix mismatch.]\label{propdelta}

If $\delta_{k-1}\leq{3\sigma^2}/{4\beta_{k}}$, then $\delta_{k}\leq 4\delta_{k-1}$.

\end{lemma}

\begin{proof}
Let $\widehat{A}_{k}={a}_{k}{a}^{\intercal}_{k}$. Hence, $\Vert \widehat{A}_{k} \Vert=\beta_{k}$. Recall that $a_k$  is the eigenvector of $\widehat{\Sigma}_{k-1}$, using the definition of $E_k = \widehat{\Sigma}_k -  \Sigma_k$, together with the recursions of the covariance matrices 
\begin{align}
\widehat{\Sigma}_k &= \widehat{\Sigma}_{k-1} - \widehat{\Sigma}_{k-1} a_k a_k^\intercal \Sigma_{k-1}/(\hat{\lambda}_k+\sigma^2), \label{recursions1}\\
\Sigma_k &= \Sigma_{k-1} - \Sigma_{k-1} a_k a_k^\intercal \Sigma_{k-1}/(a_k^\intercal\Sigma_{k-1}a_k +\sigma^2), 
\label{recursions}
\end{align} 
we have
\begin{align*}
E_{k}
&=E_{k-1}+\frac{\Sigma_{k-1} {a}_{k} {a}_{k}^{\intercal}\Sigma_{k-1}}{{a}_{k}^{\intercal}\Sigma_{k-1} {a}_{k}+\sigma^2}-\frac{\hat{\lambda}_{k}{a}_{k}{a}_{k}^{\intercal}\widehat{\Sigma}_{k-1}}{\beta_{k}\hat{\lambda}_{k}+\sigma^2}.
\end{align*}
Based on this recursion, using $\delta_k = \|E_k\|$, the triangle inequality, and Cauchy-Schwartz inequality $\|AB\| \leq \|A\|\|B\|$, we have
\begin{align*}
\delta_{k}
&\leq \delta_{k-1} +\frac{\beta_{k}\hat{\lambda}_{k}{a}_{k} E_{k-1}{a}_{k}}{(\beta_{k}\hat{\lambda}_{k}+\sigma^2)
(\beta_{k}\hat{{\lambda}}_{k}+\sigma^2-{a}_{k}^{\intercal} E_{k-1}{a})}\cdot \Vert \widehat{A}_{k}\widehat{\Sigma}_{k-1}\Vert\\
&~~~~+\frac{1}{\beta_{k}\hat{{\lambda}}_{k}+\sigma^2-{a}_{k}^{\intercal} E_{k-1}{a}_{k}}\\
&~~~~~~~~~\cdot[\hat{\lambda}_{k}(\Vert \widehat{A}_{k}E_{k-1}\Vert
+\Vert E_{k-1}\widehat{A}_{k}\Vert)
+\Vert E_{k-1}\widehat{A}_{k}E_{k-1}\Vert]\\
&\leq \delta_{k-1} +\frac{\beta_{k}^2\hat{\lambda}_{k}^2\delta_{k-1}}{(\beta_{k}\hat{\lambda}_{k}+\sigma^2)(\beta_{k}\hat{{\lambda}}_{k}+\sigma^2-\beta_{k}\delta_{k-1})}\\
&~~~~+\frac{\beta_{k}}{\beta_{k}\hat{\lambda}_{k}+\sigma^2-\beta_{k}\delta_{k-1}}[2\hat{\lambda}_{k}\delta_{k-1}+\delta_{k-1}^2]\\
&\leq (1+\frac{3\beta_{k}\hat{\lambda}_{k}}{ \beta_{k}\hat{\lambda}_{k}+\sigma^2-\beta_{k}\delta_{k-1} })\delta_{k-1}\\
&~~~~+\frac{\beta_{k}}{\beta_{k}\hat{\lambda}_{k}+\sigma^2-\beta_{k}\delta_{k-1}}\delta_{k-1}^2.
\end{align*}
Hence, if set $\delta_{k-1} \leq 3\sigma^2/(4\beta_{k})$, i.e., $\delta_{k-1}\beta_k\leq \frac{3}{4}\sigma^2$, 
 the last inequality can be upper bounded by
\[(1+3\cdot\frac{\beta_{k}\hat{\lambda}_{k}}{ \beta_{k}\hat{\lambda}_{k}+\sigma^2/4})\delta_{k-1}
+ 3\cdot\frac{\sigma^2/4}{\beta_k \hat{\lambda}_k + \sigma^2/4} \delta_{k-1}
= 4\delta_{k-1}. \]
Hence, if $\delta_{k-1} \leq 3\sigma^2/(4\beta_{k})$, we have
$\delta_{k} \leq 4\delta_{k-1}$.

\end{proof}

\begin{lemma}[Recursion for trace of the true covariance matrix]\label{proptrace}
If $\delta_{k-1}\leq \hat{\lambda}_{k}$, 
\begin{equation}
{\rm tr}(\Sigma_k)\leq {\rm tr}(\Sigma_{k-1})
-\frac{\beta_{k}\hat{\lambda}_{k}^2}{\beta_{k}\hat{\lambda}_{k}+\sigma^2}+\frac{3\beta_{k}\hat{\lambda}_{k}\delta_{k-1}}{\beta_{k}\hat{\lambda}_{k}+\sigma^2-\beta_{k}\delta_{k-1}}. \label{app_trace_bound}
\end{equation}
\end{lemma}
\begin{proof}
Let $\widehat{A}_{k}={a}_{k}{a}^{\intercal}_{k}$. 
Using the definition of $E_k$ and the recursions (\ref{recursions1}) and (\ref{recursions}), the perturbation matrix $E_k$ after $k$ iterations is given by
\begin{equation}
\begin{split}
E_k&=E_{k-1} +\hat{\lambda}_{k}^2\widehat{A}_{k}\cdot\frac{{a}_{k}^{\intercal}E_{k-1}{a}_{k}}{(\beta_{k}\hat{\lambda}_{k}+\sigma^2)(\beta_{k}\hat{\lambda}_{k}+\sigma^2-{a}_{k}^{\intercal}E_{k-1}{a}_{k})}\\
&-\frac{\hat{\lambda}_{k}}{\beta_{k}\hat{\lambda}_{k}+\sigma^2-{a}_{k}^{\intercal}E_{k-1}{a}_{k}}
\cdot(\widehat{A}_{k}E_{k-1}+E_{k-1}\widehat{A}_{k})\\
&+\frac{1}{\beta_{k}\hat{\lambda}_{k}+\sigma^2-{a}_{k}^{\intercal}E_{k-1}{a}_{k}}E_{k-1}\widehat{A}_{k}E_{k-1}.
\end{split}
\label{E_recursion}
\end{equation}
Note that ${\rm rank}(\widehat{A}_{k})=1$, thus ${\rm rank}(\widehat{A}_{k}E_{k-1})\leq 1$, therefore it has at most one nonzero eigenvalue,
\begin{align*}
\vert{\rm tr}(\widehat{A}_{k}E_{k-1})\vert
&=\vert{\rm tr}({E_{k-1}\widehat{A}_{k}})\vert \\
&=\Vert\widehat{A}_{k}E_{k-1}\Vert
\leq\Vert\widehat{A}_{k}\Vert\Vert E_{k-1}\Vert =\beta_{k}\delta_{k-1}.
\end{align*}
Note that $E_{k-1}$ is symmetric and $\hat{A}_{k}$ is positive semi-definite, we have
$
{\rm tr}(E_{k-1}\widehat{A}_{k}E_{k-1})\geq 0.
$
Hence, 
from (\ref{E_recursion}) we have
\begin{align*}
{\rm tr}(E_k)&=\rm{tr}(\widehat{\Sigma}_{k})-\rm{tr}(\Sigma_k)\\
&\geq {\rm tr}(E_{k-1})-\frac{3\beta_{k}\hat{\lambda}_{k}(\beta_{k}\hat{\lambda}_{k}+\frac{2\sigma^2}{3})\delta_{k-1}}{(\beta_{k}\hat{\lambda}_{k}+\sigma^2)(\beta_{k}\hat{\lambda}_{k}+\sigma^2-\beta_{k}\delta_{k-1})}\\
&\geq {\rm tr}(E_{k-1})-\frac{3\beta_k\hat{\lambda}_k\delta_{k-1}}{\beta_k\hat{\lambda}_k+\sigma^2-\beta_k\delta_{k-1}}.
\end{align*}
After rearranging terms we obtain
\[
{\rm tr}(\Sigma_k)\leq {\rm tr}(\Sigma_{k-1})+[{\rm tr}(\widehat{\Sigma}_k)-{\rm tr}(\widehat{\Sigma}_{k-1})]+\frac{3\beta_k\hat{\lambda}_k\delta_{k-1}}{\beta_k\hat{\lambda}_k+\sigma^2-\beta_k\delta_{k-1}}.
\]
Together with the recursion for trace of ${\rm tr}(\widehat{\Sigma}_k)$ in (\ref{tr_hat_recursion}), we have 
\begin{align*}
{\rm tr}(\Sigma_k)\leq {\rm tr}(\Sigma_{k-1})
&-\frac{\beta_{k}\hat{\lambda}_{k}^2}{\beta_{k}\hat{\lambda}_{k}+\sigma^2}+\frac{3\beta_{k}\hat{\lambda}_{k}\delta_{k-1}}{\beta_{k}\hat{\lambda}_{k}+\sigma^2-\beta_{k}\delta_{k-1}}.
\end{align*}
\end{proof}

\begin{lemma}\label{proprank}
For a given positive semi-definite matrix $X \in \mathbb R^{n\times n}$, and a vector $h\in \mathbb R^n$, if 
\[Y=X-\frac{1}{h^{\intercal}X h+\sigma^2}Xhh^{\intercal}X,\] then ${\rm rank }(X)={\rm rank}(Y).$
\end{lemma}

\begin{proof}

Apparently, for all $x\in {\rm ker}(X)$, $Yx=0$, i.e., ${\ker}(X)\subset{ \ker }(Y).$
Decompose $X=Q^{\intercal }Q$.
For all $x\in \ker (Y)$, let $b=Qh$, $z=Qx$. 
If $b=0$, $Y=X$; otherwise, when $b\neq 0$, 
we have \[0=x^{\intercal}Yx=z^{\intercal}z-\frac{z^{\intercal}bb^{\intercal}z}{b^{\intercal}b +\sigma^2}.\]
Thus,
\[ z^{\intercal} z=\frac{z^{\intercal}bb^{\intercal}z}{b^{\intercal}b +\sigma^2}\leq \frac{b^{\intercal} b}{ b^{\intercal} b+\sigma^2} z^{\intercal} z.\]
Therefore $z=0$, i.e. $x\in \ker (X)$,
$\ker (Y)\subset \ker (X).$
This shows that $\ker(X)=\ker(Y)$, or equivalently
${\rm rank}(X)={\rm rank} (Y).$
\end{proof}

\begin{proof}[Proof of Theorem \ref{entropy}]
Recall that for $k=1,\ldots, K$, $\hat{\lambda}_k\geq \chin$. Using Lemma \ref{propdelta}, we can show that for some $0<\delta<1$, if $\delta_0\leq \delta \chin/4^{K+1}\leq {3\sigma^2}/({4^{K+1}\beta_1})$, (the second inequality comes from the fact that $(1/\chin-1/\hat{\lambda}_{1})\chin\sigma^2\leq 3\sigma^2$), then for the first $K$ measurements,
\[\delta_k\leq\frac{1}{4^{K-k+1}}\frac{\delta\chin}{4}\leq\frac{1}{4^{K-k}}\frac{3\sigma^2}{4\beta_1},\quad k=1,\ldots, K.\]
Clearly, 
\[
\delta_{k-1}\leq \delta\chin/16.
\]
Hence, 
\[
(4+\delta)\delta_{k-1}\leq \delta\lambda_k.
\]
Note that $\beta_k\delta_{k-1}\leq \sigma^2$ and $\vert \lambda_k-\hat{\lambda}_k\vert\leq \delta_{k-1}$, we have 
\[
\beta_k\lambda_k\leq\beta_k(\hat{\lambda}_k+\delta_{k-1})\leq \beta_k\hat{\lambda}_k+\sigma^2.
\]
Thus,
\[
4\delta_{k-1}(\beta_k\hat{\lambda}_k+\sigma^2)+\delta \beta_k \lambda_k \delta_{k-1}\leq \delta\lambda_k(\beta_{k}\hat{\lambda}_k+\sigma^2).
\]
Then we have
\[
3\beta_k \hat{\lambda}_k \delta_{k-1}(\beta_k\hat{\lambda}_k+\sigma^2)\leq \beta_k \hat{\lambda}_k(\delta\lambda_k - \delta_{k-1})(\beta_k\hat{\lambda}_k+\sigma^2-\beta_k\delta_{k-1}),
\]
which can be rewritten as
\[
\frac{3\beta_k \hat{\lambda}_k\delta_{k-1}}{\beta_k \hat{\lambda}_k+\sigma^2-\beta_k \delta_{k-1}}\leq \frac{\beta_k \hat{\lambda}_k}{\beta_k \hat{\lambda}_k+\sigma^2}(\delta\lambda_k-\delta_{k-1}).
\]
Hence,
\[
\frac{3\beta_k \hat{\lambda}_k\delta_{k-1}}{\beta_k \hat{\lambda}_k+\sigma^2-\beta_k \delta_{k-1}}\leq \frac{\beta_k \hat{\lambda}_k}{\beta_k \hat{\lambda}_k+\sigma^2}[(\delta-1)\lambda_k+\hat{\lambda}_k],
\]
which can be written as 
\[
-\frac{\beta_k \hat{\lambda}_k^2}{\beta_k\hat{\lambda}_k+\sigma^2}+\frac{3\beta_k \hat{\lambda}_k\delta_{k-1}}{\beta_k \hat{\lambda}_k+\sigma^2-\beta_k \delta_{k-1}}\leq -(1-\delta)\frac{\beta_k\hat{\lambda}_k}{\beta_k \hat{\lambda}_k+\sigma^2}\lambda_k.
\]
By applying Lemma \ref{proptrace}, we have
\begin{align*}
&{\rm tr}(\Sigma_k)
\leq {\rm tr}{(\Sigma_{k-1})}-(1-\delta)\frac{\beta_{k}\hat{\lambda}_{k}}{\beta_{k}\hat{\lambda}_{k}+\sigma^2}\lambda_{k}\\
&\leq {\rm tr}(\Sigma_{k-1})-(1-\delta)\frac{\beta_{k}\hat{\lambda}_k}{\beta_{k}\hat{\lambda}_{k}+\sigma^2}\frac{{\rm tr}(\Sigma_{k-1})}{s} \triangleq  f_{k} {\rm tr}(\Sigma_{k-1}),
\end{align*}
where we have used the definition for $f_k$ in (\ref{def_f}).
Subsequently, 
\[{\rm tr}(\Sigma_k)\leq (\prod_{j=1}^{k}f_j) {\rm tr}(\Sigma_0).\]
Lemma \ref{proprank} shows that the rank of the covariance will not be changed by updating the covariance matrix sequentially:
${\rm rank} (\Sigma_1)=\cdots={\rm rank} (\Sigma_k)=s$. Hence, we may decompose the covariance matrix $\Sigma_k=Q Q^{\intercal}$, with $Q\in \mathbb R^{n\times s}$ being a full-rank matrix, then
${\textsf {Vol}}(\Sigma_k)={\rm det}(Q^{\intercal} Q).$
Since ${\rm tr}(Q^{\intercal} Q)={\rm tr}(Q Q^{\intercal})$, we have
\begin{align*}
{\textsf {Vol}}^2(\Sigma_k)&={\rm det}{(Q^{\intercal} Q)}
\overset{(1)}{\leq} \prod_{j=1}^{s}(Q^{\intercal} Q)_{jj}\\
&\overset{(2)}{ \leq} \left(\frac{{\rm tr}( Q^{\intercal} Q)}{s}\right)^s
=\left(\frac{{\rm tr(\Sigma_k)}}{s}\right)^s,
\end{align*}
where (1) follows from the Hadamard's inequality and (2) follows from the mean inequality.
Finally, we can bound the conditional entropy of the signal as 
\begin{align*}
\entropy[y_{j}, a_{j}, j \leq k]{x} 
&= \ln (2\pi e)^{s/2} {\textsf{Vol}}(\Sigma_k) \\
&\leq \frac{s}{2}\ln \{2\pi e (\prod_{j=1}^{k}f_j) {\rm tr}(\Sigma_0)\},
\end{align*}
which leads to the desired result.
\end{proof}


\begin{proof}[Proof of Theorem \ref{thm:power}]
Recall that ${\rm rank}(\Sigma)=s$, and hence $\lambda_k=0$, $k = s+1, \ldots, n$.
Note that for each iteration, the eigenvalue of $\widehat{\Sigma}_k$ in the direction of $a_k$, which corresponds to the largest eigenvalue of $\widehat{\Sigma}_k$, is eliminated below the threshold $\chin$. Therefore, as long as the algorithm continues, the largest eigenvalue of $\widehat{\Sigma}_k$ is exactly the $(k+1)$th largest eigenvalue of $\widehat{\Sigma}$.
Now if 
\begin{equation}\delta_0\leq \chin/4^{s+1},
\label{assumption_delta}
\end{equation}
 using Lemma \ref{lemmaeig} and Lemma \ref{propdelta}, we have that
$$\vert \hat{\lambda}_{k}-\lambda_k \vert\leq \delta_0,\ {\rm for}\ k=1,\ldots,s,$$
$$\vert \hat{\lambda}_{j} \vert\leq \delta_0\leq  \chin-\delta_{s},\ {\rm for}\ k=s+1,\ldots,n.$$
In the ideal case without perturbation, each measurement decreases the eigenvalue along a given eigenvector to be below $\chin$.
Suppose in the ideal case, the algorithm terminates at $K\leq s$ iterations, which means $$\lambda_1 \geq\cdots\geq \lambda_L\geq\chin >\lambda_{K+1}(\Sigma)\geq\cdots\geq\lambda_{s}(\Sigma),$$
and the total power needed is
\begin{equation}
P_{\rm ideal}=\sum_{k=1}^{K}\sigma^2\left(\frac{1}{\chin}-\frac{1}{\lambda_k}\right).
\label{P_ideal}
\end{equation}

On the other hand, in the presence of perturbation, the algorithm will terminate using more than $K$ iterations since with perturbation, eigenvalues of $\Sigma$ that originally below $\chin$ may get above $\chin$. In this case, we will also allocate power while taking into account the perturbation: 
\[\beta_k=\sigma^2\left(\frac{1}{\chin-\delta_s}-\frac{1}{\hat{\lambda}_{k}}\right).\]
This suffices to eliminate even the smallest eigenvalue to be below threshold $\chin$ since
$$\frac{\sigma^2\hat{\lambda}_{k-1}}{\beta_{k-1}\hat{\lambda}_{k-1}+\sigma^2}= \chin-\delta_s < \chin.$$
We first estimate the total amount of power used at most to eliminate eigenvalues $\hat{\lambda}_k$, for $K+1\leq k\leq s$: 
\begin{align*}
\beta_k
&=\sigma^2(1/(\chin-\delta_s) - 1/\hat{\lambda}_k)\\
& \leq \sigma^2(1/(\chin-\delta_s)- 1/(\chin+\delta_0))\\
&\leq \sigma^2 \frac{(4^s+1)\delta_0}{(\chin-4^s \delta_0)(\chin+\delta_0)} \leq \frac{20}{51}\frac{\sigma^2}{\chin}.
\end{align*}
where we have used the fact that $\delta_s \leq 4^s\delta_0$ (a consequence of Lemma \ref{propdelta}), the assumption (\ref{assumption_delta}), and monotonicity of the upper bound in $s$. 
The total power to reach precision $\varepsilon$ in the presence of mismatch can be upper bounded by
\begin{align*}
&P_{\rm mismatch}
\leq \sum_{k=1}^{s}\beta_k\\
&\leq\sigma^2\left\{\sum_{k=1}^{K}\left(\frac{1}{\chin-\delta_s}-\frac{1}{\hat{\lambda}_{k}}\right)+\frac{20(s-K)}{51}\frac{\sigma^2}{\chin}\right\}.
\end{align*}
In order to achieve precision $\varepsilon$ and confidence level $p$,
the extra power needed is upper bounded as
\begin{align*}
&P_{\rm mismatch}-P_{\rm ideal}\\
&\leq \sigma^2\left\{\sum_{k=1}^{K} \left(\frac{1}{3}\frac{1}{\chin}+\frac{\delta_0}{\lambda_k^2}\right) +\frac{20(s-K)}{51}\frac{1}{\chin}\right\}\\
&\leq\sigma^2\left\{\frac{1}{4^{s+1}}\sum_{k=1}^{K}\frac{\chin}{\lambda_k^2}+\frac{20s-3K}{51}\frac{1}{\chin}\right\} \\
&<\left(\frac{20}{51}s-(\frac{3}{51}-\frac{1}{4^{s+1}})K\right)\frac{\sigma^2}{\chin}\\
&\leq \left(\frac{20}{51}s+\frac{1}{272}K\right)\frac{\sigma^2}{\chin},
\end{align*}
where we have again used $\delta_s\leq 4^s \delta_0 \leq 4^s \chin/4^{s+1} = \chin/4$, $1/\hat{\lambda}_k - 1/\lambda_k \leq \delta_0/\lambda_k^2$, the fact that $\lambda_k \geq \chin$ for $k = 1, \ldots, K$.
\end{proof}

\begin{proof}[Proof of Lemma \ref{cor:sampleSize}]
It is a direct consequence of Lemma \ref{zhu2012}. 
Let $\theta={\rm tr}(\Sigma)/\Vert\Sigma\Vert\geq 1$. 
For some constant $\delta > 0$, set
\[L\geq 4n^{1/2}{\rm tr}(\Sigma)({\Vert\Sigma\Vert}/{\delta^2}+{4}/{\delta}).\]
%
Then from Lemma \ref{zhu2012}, we have
\begin{align*}
& P\{\Vert\widehat{\Sigma}-\Sigma\Vert \leq \delta\}\\
& \geq  P\{\Vert\widehat{\Sigma}-\Sigma\Vert \leq \left(\sqrt{2n^{1/2}(\theta+1)/L}+2\theta n^{1/2}/L\right)\Vert\Sigma\Vert\}\\
&> 1-2n\exp(-\sqrt{n}).
\end{align*}
\end{proof}

The following Lemma is used in the proof of  Lemma \ref{thm:cov-sketch}. 
\begin{lemma}\label{est}
For the setup in Section \ref{app:cov_sketch}, if for some constant $M$, $N$ 
and $L$ satisfies the conditions in Lemma \ref{thm:cov-sketch}, 
then $\Vert\eta\Vert_1\leq \tau$ with probability exceeding $1-{2}/{n}-{2}/{\sqrt{n}}-2n\exp(-c_1M)$ for some universal constant $c_1>0$.
\end{lemma}
\begin{proof}
Let $\theta = {\rm tr}(\Sigma)/\|\Sigma\|$. From Chebyshev's inequality, we have that
\[\mathbb P\{\vert z_i\vert<\frac{\tau}{6M}\} \geq 1-\frac{36M^2\sigma^2{\rm tr}(\Sigma)}{NL\tau^2}, \quad i = 1, \ldots, K\] 
\[\mathbb P\{|w|<M\frac{\sigma^2}{L}+\frac{\tau}{6}\} \geq 1-\frac{72\sigma^4M}{NL^2\tau^2},\] and 
\[\mathbb P\{|b|<(M+\sqrt{M})n\} \geq 1-\frac{2}{n}.\] 
When
\begin{equation}
N \geq 4n^{1/2}{\rm tr}(\Sigma)(\frac{36n^2M^2\Vert\Sigma\Vert}{\tau^2}+\frac{24nM}{\tau}), \label{N_lower_bound}
\end{equation}
with the concentration inequality for Wishart distribution in Lemma \ref{zhu2012} and plugging in the lower bound for $N$ in (\ref{N_lower_bound}) and the definition for $\tau$ in (\ref{def_tau}) we have
\begin{align*}
&\mathbb P\{\Vert\widehat{\Sigma}_N-\Sigma\Vert \leq {\tau}/[{3n(M+\sqrt{M})}]\}\\
\geq &~ \mathbb P\{\Vert\widehat{\Sigma}_N-\Sigma\Vert \leq (\sqrt{\frac{2n^{1/2}\theta}{N}}+\frac{2\theta n^{1/2}}{N})\Vert\Sigma\Vert\}\\
>&~ 1-2n\exp(-\sqrt{n}). 
\end{align*}
Furthermore, when $L$ satisfies (\ref{L_LB}), 
we have 
\begin{align*}
&\mathbb P\{\vert z_i\vert<\frac{\tau}{6M}\}
\geq 1-\frac{1}{M\sqrt{n}},\\
&\mathbb P\{|w|<\frac{\tau}{3}\}
\geq 1-\frac{1}{\sqrt{n}},\\
&\mathbb P\{\vert b\vert<(M+\sqrt{M})n\}
\geq 1-\frac{2}{n}.
\end{align*}
Therefore, $\Vert\eta\Vert_1\leq \tau$
holds with probability at least $1-{2}/{n}-{2}/{\sqrt{n}}-2n\exp(-\sqrt{n})$.
\end{proof}

\begin{proof}[Proof of Lemma \ref{thm:cov-sketch}]
With Lemma \ref{est}, let $\tau={M\delta}/c_2$,
the choices of $M$, $N$, and $L$ ensure that $\Vert\eta\Vert_1\leq{M \delta}/c_2$ with probability at least $1-{2}/{n}-2/\sqrt{n}-2n\exp(-\sqrt{n})$.
By  Lemma \ref{sketching} and noting that the rank of $\Sigma$ is $s$, we have
$\Vert \widehat{\Sigma}-\Sigma \Vert_F\leq \delta.$
Therefore, with probability exceeding $1-2/n-{2}/{\sqrt{n}}-2n\exp(-\sqrt{n})-\exp(-c_0c_1ns),$
\[
\Vert \widehat{\Sigma}-\Sigma \Vert\leq \Vert \widehat{\Sigma}-\Sigma \Vert_F\leq \delta.
\]
\end{proof}
The proof of Theorem \ref{lemma:GMM_mismatch} will use the following two lemmas. 
\begin{lemma}[Moment generating function of multivariate Gaussian \cite{VincentWakin}] Assume $X\sim \mathcal N(0,\Sigma)$. The moment generating function of $\Vert X \Vert_2$ is 
\[
\mathbb E[e^{s\Vert X \Vert_2}]=1/\sqrt{I-2s\Sigma}.
\]
\end{lemma}

\begin{proof}[Proof of Theorem \ref{lemma:GMM_mismatch}]
We adapt the technique used in \cite{InfoGreedy2014} for proving performance bound for GMM signal without mismatch. Suppose the true signal is generated from the $c^*$th component. 
First,  apply measurements to each component $c\in [C]$. Clearly, spending a total amount of power $\sum_{c=1}^{C}m_c$ would suffice to ensure that the norm of covariance of each individual component is below $\chin$.
In the ideal case, the weight is updated in the following manner:
\[
\pi_{c}^{k+1}=\pi_c^{k}L_k\exp\{-\frac{1}{2}\frac{(y_k-a_k^{\intercal}\mu_{c,k-1})^2}{a_k^{\intercal}\Sigma_{c,k-1}a_k+\sigma^2}\}, \quad c \in [C].
\]
In the presence of mismatch, this becomes 
\[
\hat{\pi}_{c}^{k+1}=\hat{\pi}_c^{k}\hat{L}_k\exp\{-\frac{1}{2}\frac{(y_k-a_k^{\intercal}\hat{\mu}_{c,k-1})^2}{a_k^{\intercal}\widehat{\Sigma}_{c,k-1}a_k+\sigma^2}\}, \quad c \in [C].
\]
The $L_k$ and $\hat{L}_k$ for $k=1,2,\ldots$ are normalization coefficients. After $m$ measurements,
\begin{align*}
&~~~~\frac{1}{m} \left\vert  \sum_{k=1}^{m} \sum_{\ell=1}^{C}\frac{(y_{k}-a_{k}^{\intercal}\hat{\mu}_{\ell,k-1})^2}{a_k^{\intercal}\widehat{\Sigma}_{\ell,k-1}a_k+\sigma^2}\cdot\hat{\pi}_{\ell}^{k}-\frac{(y_k-a_k^{\intercal}\mu_{c^*,k-1})^2}{a_k^{\intercal}\Sigma_{c^*,k-1}a_k+\sigma^2} \right\vert\\
&=\frac{1}{m}\left\vert \sum_{k=1}^{m} \sum_{\ell=1}^{C}\frac{(y_k-a_k^{\intercal}\hat{\mu}_{\ell,k-1})^2}{a_k^{\intercal}\widehat{\Sigma}_{\ell,k-1}a_k+\sigma^2}\cdot\hat{\pi}_{\ell}^{k} - \sum_{k=1}^m \frac{(y_k-a_k^{\intercal}\hat{\mu}_{c^*,k-1})^2}{a_k^{\intercal}\widehat{\Sigma}_{c^*,k-1}a_k+\sigma^2}\right.\\
&~~~~~~~
\left.+\sum_{k=1}^{m}(\frac{(y_k-a_k^{\intercal}\hat{\mu}_{c^*,k-1})^2}{a_k^{\intercal}\widehat{\Sigma}_{c^*,k-1}a_k+\sigma^2}-\frac{(y_k-a_k^{\intercal}\mu_{c^*,k-1})^2}{a_k^{\intercal}\Sigma_{c^*,k-1}a_k+\sigma^2})\right\vert\\
&\leq \hat{\eta}+\frac{2{\rm ln}\vert C\vert}{m}\\
&~~~~+\frac{1}{m}\sum_{k=1}^{m}\left\vert\frac{(y_k-a_k^{\intercal}\hat{\mu}_{c^*,k-1})^2}{a_k^{\intercal}\widehat{\Sigma}_{c^*,k-1}a_k+\sigma^2}-\frac{(y_k-a_k^{\intercal}\mu_{c^*,k-1})^2}{a_k^{\intercal}\Sigma_{c^*,k-1}a_k+\sigma^2}\right\vert.\\
\end{align*}
Now we study bound for each individual term inside the sum over $k$. To simplify notation, we omit the dependence on $k$, $c^*$ and $k-1$ without causing confusion.  
In the following let $z \triangleq y - a^\intercal\mu = a^\intercal(x-\mu) + w$, and let $\varrho\triangleq a^\intercal (\hat{\mu}-\mu)$. Hence, $|\varrho| \leq |\beta|\cdot|\hat{\mu}-\mu|$
is bounded.
Note that
\begin{equation}
\begin{split}
&~~~~\left\vert \frac{(y -a^{\intercal}\hat{\mu})^2}{a^{\intercal}\widehat{\Sigma}a+\sigma^2}-\frac{(y-a^{\intercal}\mu)^2}{a^{\intercal}\Sigma a+\sigma^2} \right\vert\\
&\leq \left\vert \frac{(y-a^{\intercal}\hat{\mu})^2}{a^{\intercal}\widehat{\Sigma}a+\sigma^2} - \frac{(y-a^{\intercal}\mu)^2}{a^{\intercal}\widehat{\Sigma}a+\sigma^2} \right\vert+\left\vert \frac{(y-a^{\intercal}{\mu})^2}{a^{\intercal}\widehat{\Sigma}a+\sigma^2} - \frac{(y-a^{\intercal}\mu)^2}{a^{\intercal}{\Sigma}a+\sigma^2}\right\vert\\
&\leq \frac{2|\varrho|\cdot|a^{\intercal}(x-\mu)+w-\varrho/2|}{\sigma^2}+ (y-a^{\intercal}\mu)^2 \vert\frac{\beta\delta}{\sigma^4}\vert
\\
& = \frac{2|\varrho|\cdot|z-\varrho/2|}{\sigma^2} + |z|^2 \frac{\beta\delta}{\sigma^4}\\
&\leq \frac{\vert \varrho \vert^2+2\vert\varrho\vert \vert z\vert}{\sigma^2}+|z|^2 \frac{\beta\delta}{\sigma^4}.
\label{last}
\end{split}
\end{equation}
%
Since $m\leq n$, from (\ref{last}) we have 
\begin{align*}
&\max_{k=1}^m \left\vert\frac{(y_k-a_k^{\intercal}\hat{\mu}_{c^*,k-1})^2}{a_k^{\intercal}\widehat{\Sigma}_{c^*,k-1}a_k+\sigma^2}-\frac{(y_k-a_k^{\intercal}\mu_{c^*,k-1})^2}{a_k^{\intercal}\Sigma_{c^*,k-1}a_k+\sigma^2}\right\vert\\
&\leq\max_{k=1}^n \left\vert\frac{(y_k-a_k^{\intercal}\hat{\mu}_{c^*,k-1})^2}{a_k^{\intercal}\widehat{\Sigma}_{c^*,k-1}a_k+\sigma^2}-\frac{(y_k-a_k^{\intercal}\mu_{c^*,k-1})^2}{a_k^{\intercal}\Sigma_{c^*,k-1}a_k+\sigma^2}\right\vert\\
&\leq \frac{1}{\sigma^4}\cdot \max_{k=1}^n \left \{\sigma^2(\vert\varrho_k\vert+2\vert z_k \vert)\vert\varrho_k\vert +\beta_k\vert z_k\vert^2\delta_{k-1}\right\}.
\end{align*}

Note that $z_{k} \triangleq a_k^\intercal(x-\mu_{c^*, k-1}) + w_k \sim \mathcal{N}(a_k^\intercal(\mu_{c^*} - \mu_{c^*,k-1}), a_k^\intercal\Sigma_{c^*,k-1}a_k + \sigma^2)$, so for some $t \in (0, 1)$, we have
\begin{align*}
&|z_{k}|\\
&< \frac{1}{t}\sqrt{[(\lambda_{c^*,k} + \delta_{k-1}) \beta_k + \sigma^2] + (a_k^\intercal(\mu_{c^*} - \mu_{c^*,k-1}))^2}
\triangleq  \frac{b_k}{t}
\end{align*}
 with probability exceeding $1-t^2$, where $b_k$ is bounded.

Finally, 
\begin{align*}
&~~~~\frac{1}{m} \left\vert  \sum_{k=1}^{m} (\sum_{\ell=1}^{C}\frac{(y_{k}-a_{k}^{\intercal}\hat{\mu}_{\ell,k-1})^2}{a_k^{\intercal}\widehat{\Sigma}_{\ell,k-1}a_k+\sigma^2}\cdot\hat{\pi}_{\ell}^{k}-\frac{(y_k-a_k^{\intercal}\mu_{c^*,k-1})^2}{a_i^{\intercal}\Sigma_{c^*,k-1}a_k+\sigma^2})\right \vert\\
&\leq \hat{\eta}+\frac{2{\rm ln}\vert C \vert}{m} 
+ \underbrace{\frac{1}{\sigma^4}\cdot \max_{k=1}^n \left \{\sigma^2(\vert\varrho_k\vert+2\frac{b_k}{t})\vert\varrho_k\vert +\beta_k\frac{b_k^2}{t^2}\delta_{k-1}\right\}}_{\eta_0}
\end{align*}
with probability at least $1-nt^2$. Let 
\[
\Delta=\max_{k=1}^{n} \left\{\max(\varrho_k, \delta_{k-1})\right\},
\] 
and let 
\[U=\frac{1}{\sigma^4}\cdot \max_{k=1}^n \left\{ \sigma^2(\Delta+2 n b_k)+\beta_kn^2b_k^2 \right\}.
\]
We choose $t = 1/n$. Then 
\[
\eta_0=U\Delta.
\]
Note that when there is no mismatch, $\delta_{k-1}=\varrho_k=0$ for $k\in[n]$, which leads to $\Delta=0$ and thus $\eta_0=0$.
Here $\hat{\eta} = 1/2$ is a parameter used in the multiplicative weight update algorithm. In particular, we can identify the correct component $c^*$ with probability $1-1/n$ whenever $m=\mathcal{O}\left({\rm ln}\vert C \vert/(\hat{\eta}+\eta_0)\right)$.
For $k=1,\ldots,m$, we choose $\beta_k=1$. Thus, we need at most
\[
\sum_{c=1}^{C} m_c +\mathcal{O}\left(\frac{{\rm ln}\vert C \vert}{\hat{\eta}+\eta_0}\right)
\]
amount of power in total.
\end{proof}

Note that $|\varrho_{k}|$ can be computed recursively. We may derive a recursion. Let $z_k \triangleq a_k^\intercal(x-\mu_{k-1}) + w_k = y_k - a_k^\intercal \mu_{k-1}$. Also Let $\varrho_k \triangleq a^{\intercal}(\hat{\mu}_k-\mu_k)$. Note that $\varrho_k = a^\intercal \xi_k$ for $\xi_k = \hat{\mu}_k-\mu_k$ in (\ref{recursion_xi}). Based on the recursion for $\xi_k$ in (\ref{recursion_xi}) that we derived earlier, we have
\[\varrho_k
=\frac{\sigma^2}{\beta_k\hat{\lambda}_k+\sigma^2}[\varrho_{k-1}+\frac{a_k^\intercal E_{k-1} a_k (y_k-a_k^{\intercal}\mu_{k-1})}{\beta_k\hat{\lambda}_k+\sigma^2-a^{\intercal}_kE_{k-1}a_k}]\]
and 
\[|\varrho_k|
\leq \frac{1}{\hat{\lambda}_k (\beta_k/\sigma^2)+1}[|\varrho_{k-1}|+\frac{\delta_k }{(\hat{\lambda}_k-\delta_k)+\sigma^2/\beta_k} |z_k|].\]


%
\begin{proof}[Proof of Lemma \ref{lemma:one_sparse}]
The recursion of the diagonal entries can be written as
\begin{align*}
\Sigma_{ii}^{(k)}
&=\Sigma_{ii}^{(k-1)}-\frac{(\Sigma_{i j_{k-1}}^{(k-1)})^2}{\Sigma_{j_{k-1} j_{k-1}}^{(k-1)}+\sigma^2/\beta_k}\\
&=\frac{\Sigma_{ii}^{(k-1)}\Sigma_{j_{k-1} j_{k-1}}^{(k-1)}(1-\rho^{(k-1)}_{i j_{k-1}})+\Sigma_{ii}^{(k-1)}\sigma^2/\beta_k}{\Sigma^{(k-1)}_{j_{k-1} j_{k-1}}+\sigma^2/\beta_k}.
\end{align*}
Note that for $i=j_{k-1}$, 
\[
\Sigma_{j_{k-1}j_{k-1}}^{(k)}=\frac{\Sigma_{j_{k-1}j_{k-1}}^{(k-1)}\sigma^2/\beta_k}{\Sigma^{(k-1)}_{j_{k-1} j_{k-1}}+\sigma^2/\beta_k}\leq \frac{\gamma}{1+\gamma}\Sigma_{j_{k-1}j_{k-1}}^{(k-1)},
\]
and for $i\neq j_{k-1}$,
\begin{align*}
\Sigma_{ii}^{(k)}
&\leq \frac{\Sigma_{ii}^{(k-1)}\Sigma_{j_{k-1} j_{k-1}}^{(k-1)}(1-\rho^{(k-1)})+\Sigma_{ii}^{(k-1)}\sigma^2/\beta_k}{\Sigma^{(k-1)}_{j_{k-1} j_{k-1}}+\sigma^2/\beta_k}\\
&\leq \Sigma_{ii}^{(k-1)}\frac{\Sigma_{j_{k-1} j_{k-1}}^{(k-1)}(1-\rho^{(k-1)})+\sigma^2/\beta_k}{\Sigma_{j_{k-1} j_{k-1}}^{(k-1)}+\sigma^2/\beta_k}\\
&\leq \Sigma_{ii}^{(k-1)}\frac{1-\rho^{(k-1)}+\gamma}{1+\gamma}.
\end{align*}
Therefore,
\begin{align*}
{\rm tr}(\Sigma_k)
&\leq (1-\frac{\rho^{(k-1)}}{1+\gamma}){\rm tr}(\Sigma_{k-1})-\frac{1-\rho^{(k-1)}}{1+\gamma}\Sigma_{j_{k-1}j_{k-1}}^{(k-1)}\\
&\leq [1-\frac{(n-1)\rho^{(k-1)}+1}{n(1+\gamma)}]{\rm tr}(\Sigma_{k-1}).
\end{align*}
\end{proof}

\begin{proof}[Proof of Theorem \ref{Gaussian_one_sparse}]
Let $\varepsilon\geq\sqrt{\|\Sigma_K\|\cdot\chi_n^2(p)}$, i.e. $\|\Sigma_K\|\leq \chin$. 
Then Theorem \ref{Gaussian_one_sparse} follows from   
\begin{align}
&\mathbb{P}_{x\sim \mathcal{N}(\mu_K, \Sigma_K)}[\|x-\mu_K\|_2\leq \varepsilon]
\nonumber\\
&\geq \mathbb{P}_{x\sim \mathcal{N}(\mu_K, \Sigma_K)}
[\|x-\mu_K\|_2\leq \sqrt{\|\Sigma_K\|\cdot\varepsilon^2}]\nonumber\\
&\geq \mathbb{P}_{x\sim \mathcal{N}(\mu_K, \Sigma_K)}
[(x-\mu_K)^\intercal{\Sigma_K}^{-1}(x-\mu_K)\leq \chi_n^2(p)] = p.\label{app1}
\end{align}
From Lemma \ref{lemma:one_sparse}, we have that when the powers $\beta_i$ are sufficiently large
\[
\Vert\Sigma_K\Vert\leq {\rm tr}(\Sigma_K)\leq (1-\frac{1}{n(1+\gamma)})^K{\rm tr}(\Sigma).
\]
Hence for (\ref{app1}) to hold, we can simple require $(1-\frac{1}{n(1+\gamma)})^K{\rm tr}(\Sigma) \leq \chin$, or equivalently (\ref{K_one_sparse}) in Theorem \ref{Gaussian_one_sparse}.
\end{proof}

\end{document}